\documentclass[journal]{IEEEtran}
\usepackage{graphicx}
\usepackage{amsmath}
\usepackage{amsfonts}
\usepackage{amssymb}
\usepackage{bm}
\usepackage{subfigure}
\usepackage{threeparttable,array}
\usepackage{color}
\usepackage[ruled,linesnumbered]{algorithm2e}
\usepackage{amsthm}

\newtheorem{theorem}{Theorem}
\newtheorem{remark}{Remark}
\newtheorem{definition}{Definition}
\newtheorem{lemma}{Lemma}

\newtheorem{corollary}{Corollary}

\newcounter{MYtempeqncnt}

\ifCLASSINFOpdf
\else
\fi
\begin{document}
\graphicspath{{figures/}}
%
\title{Roto-Translation Invariant Formation of Fixed-Wing UAVs in 3D: Feasibility and Control}
%
%
%

\author{Xiaodong~He,
        Zhongkui~Li,
        Xiangke~Wang, and 
        Zhiyong Geng
\thanks{This work was supported by the National Natural Science Foundation of China (61973006, T2121002), and the China Postdoctoral Science Foundation (2022M720242).
\emph{(Corresponding author: Zhongkui Li.)}}
\thanks{X. He, Z. Li, and Z. Geng are with the State Key Laboratory for Turbulence and Complex Systems, Department of Mechanics and Engineering Science, College of Engineering, Peking University, Beijing 100871, China (e-mail: hxdupc@pku.edu.cn; zhongkli@pku.edu.cn; zygeng@pku.edu.cn)

X. Wang is with the College of Intelligence Science and Technology, National University of Defense Technology, Changsha 410073, China (xkwang@nudt.edu.cn)}
}

\maketitle

\begin{abstract}
This paper investigates the formation of fixed-wing UAVs in 3D, which communicate via a directed acyclic graph. Different from common formation problems, we consider the roto-translation invariant (RTI) formation, where the ``roto-translation" refers to a rigid-body motion obtained by composing rotation and translation. Besides, the fixed-wing UAV is modelled by a 3-D nonholonomic constrained rigid body instead of a particle agent. The main results of this paper include proposing the formation feasibility conditions and designing the formation controller. Firstly, we define the RTI formation and propose the conditions to guarantee that the formations are feasible for the fixed-wing UAVs under the nonholonomic and input saturation constraints. Secondly, given feasible formations, we design a formation controller by introducing a virtual leader and employing the compensation of rotation, followed by proving the stability of the closed-loop system. Finally, simulation examples are presented to verify the theoretical results.al with the motion coupling among multiple robots. Finally, numerical simulations are conducted to verify the theoretical results.
\end{abstract}


%
\IEEEpeerreviewmaketitle

\section{Introduction}

Fixed-wing UAVs are playing increasingly prominent roles in defense programs and civil applications, such as border patrol, disaster relief, and environmental monitoring \cite{Beard2012Small,Beard2006Decentralized,Sujit2014Unmanned,Wang2019Coordinated,Liu2022Mission}. Among various coordination tasks, the most common one for fixed-wing UAVs is the formation control \cite{Oh2015A_survey}, which aims at driving multiple UAVs to achieve the predefined constraints regarding their states, generally interpreted as the desired geometric formation patterns. Plenty of existing articles have investigated the formation control of fixed-wing UAVs with diverse approaches, such as leader-follower structure \cite{Wang2021Distributed,Zhang2021Robust}, consensus-based method \cite{Muslimov2021Consensus,Guzey2017Hybrid}, guiding vector field \cite{Wang2022A_Fixed-Wing,Baldi2022Adaptation}, graph rigidity \cite{Sun2019Circular,Bayezit2013Distributed}. Additionally, several papers further focus on the formation reconfiguration \cite{Chen2021Formation,Wang2021Formation} and maneuverability \cite{Challa2021On_Maneuverability,Wang2020Motion}. Although having been extensively studied, the formation of fixed-wing UAVs still remains some implicit but attractive facts which are typically omitted.

One fact is that the model of the fixed-wing UAV is a rigid body in 3D, indicating that the complete 3-D rigid body kinematics should be established to capture in mathematics the important characteristics of the physical system. However, in a large amount of research \cite{Wang2021Distributed,Muslimov2021Consensus,Sun2019Circular,Chen2021Formation,Wang2020Motion,Sun2021Collaborative,Chen2020On_Dubins}, the fixed-wing UAV is modelled by a planar nonholonomic vehicle, generally known as a unicycle. Although the model of a 2-D nonholonomic vehicle coincides with the cases of flying at a fixed altitude, the results obtained in 2D cannot be directly applied to the formation problems of the fixed-wing UAVs in 3D.

Another fact in formation control is that it is more appealing in certain real-world scenarios to focus on the overall motion of a formation pattern. It has been demonstrated in \cite{Das2002Vision-Based} that the concept of formation control should involve not only controlling the relative positions and orientations of the agents in a group, but also allowing the group to move as a whole. Nevertheless, only a few papers \cite{Li2016Notion,Peng2020Mobile,Chen2022Three-Dimensional} investigate the formation problem from the perspective of overall motion. Although providing some insightful results, these articles omitted the formation feasibility, which is indeed the premise of control especially under certain motion constraints.

Generally, it is said that a formation pattern is feasible, if all of the agents in the formation satisfy the motion constraints, such as nonholonomic constraints, input saturation constraints, and so on. The seminal paper \cite{Tabuada2005Motion} develops a systematic framework for studying the motion feasibility of multi-agent formations, where feasibility conditions are proposed to maintain formation specifications described by equality constraints. The work by Sun and Anderson \cite{Sun2016Formation} investigates the formation feasibility of heterogeneous robot teams modelled by control affine nonlinear systems, particularly presenting
two special motions that preserve motion constraints. The paper \cite{Morbidi2018A_New} characterizes the mobility for formations of unicycle robots, where the feasibility is formulated based on the distance-bearing constraints. Colombo and Dimarogonas \cite{Colombo2020Motion} have studied the motion feasibility problem for multi-agent systems on Lie groups, especially for the left-invariant fully actuated kinematic and dynamic systems. Besides, the dissertation \cite{Whitzer2020Coordinating} addresses the multi-robot planning problem with motion constraints, and develops an algorithm for fixed-wing UAVs in formation flight.

Motivated by the above-mentioned literature, in this paper, we investigate the roto-translation invariant (RTI) formation feasibility and control of fixed-wing UAVs in 3D under the communication topology of a directed acyclic graph. The terminology of ``roto-translation invariant" is proposed in \cite{Consolini2012On}, where roto-translation describes a rigid-body motion obtained by composing rotation and translation. Thus, the RTI formation can rotate and translate simultaneously, and meanwhile the shape of the formation pattern keeps invariant, moving like a single rigid body. The contributions of this paper are threefold.

Firstly, we establish the kinematics of the fixed-wing UAV, which is a general 3-D rigid body under the nonholonomic constraints, rather than a 3-D nonholonomic particle agent with decoupled rotation \cite{Yao2020Path,Zhao2018Affine}. Different from the 2-D rotation which has only one degree of freedom (DOF), the rotation kinematics in 3D is not an integrator anymore but has a nonlinear formulation with the coupling of three DOFs. This leads to the rigid body's attitude being related to the rotation sequence, which brings essential challenges than the 2-D case as given in \cite{He2021Roto}. Particularly, in order to globally and uniquely depict the rotation, the fixed-wing UAV's kinematics is established in the Lie group ${\rm SE(3)}$ \cite{Verginis2019Robust,Thunberg2016Consensus,Wang2022Hybrid,Liu2021Geometry,Wang2012A_Dual}, avoiding the singularity and unwinding phenomenon resulting from Euler angles.

Secondly, we define the RTI formation of the fixed-wing UAVs, and propose the formation feasibility under the nonholonomic and input saturation constraints. Since the RTI formation possesses special motion characteristics, not all formation patterns are feasible to keep RTI. Thus, we propose the formation feasibility by deriving the condition of formation maintenance, and embedding the nonholonomic and saturation constraints into it. It is shown that the relative positions in the formation can be arbitrarily specified, while the relative orientations are decided by the nonholonomic constraints. Moreover, it is derived mathematically that when the RTI formation turns around, the outside fixed-wing UAVs accelerate while the inside ones decelerate.

Thirdly, based on the feasible formations, we design the formation controller for the fixed-wing UAVs. Note that the fixed-wing UAV is an underactuated system in the sense that six DOFs are controlled by only four inputs. Hence, problem, we present the idea of feedback coupling for underactuated systems. Specifically, the underactuated states are made coupled with the actuated directions by designing appropriate feedback in the controller, so that the lack of control inputs can be compensated by the additional constructed terms. Moreover, these additional feedback are designed from a physical perspective, representing the rotation from the constrained velocity to the standard velocity. Then, intuitively, the fixed-wing UAVs move along the trajectories of fully actuated systems to achieve the formation.

The paper is organized as follows. Section~\ref{sec_prelimi} provides the preliminaries and the problem statement, followed by the communication topology transformation in Section~\ref{sec_topolo}. Section~\ref{sec_RTI_feasi} and Section~\ref{sec_control} propose the formation feasibility and the formation controller, respectively. Numerical simulation examples are given in Section~\ref{sec_simulation}, and the paper concludes in Section~\ref{sec_conclusion}.

\section{Preliminaries and problem statement} \label{sec_prelimi}

\subsection{Kinematics of fixed-wing UAV}

The fixed-wing UAV is modelled by a 6-DOF rigid body moving in $\mathbb{R}^3$. Let $\bm{\mathcal{F}}_{\rm e}$ denote the earth-fixed frame and $\bm{\mathcal{F}}_{\rm b}$ the body-fixed frame. The position is described by a vector $\bm{p}=[x\ \ y\ \ z]^{T}$, while the attitude is specified by a rotation matrix $\bm{R}$ in the Special Orthogonal group ${\rm SO(3)}=\{\bm{R}\in\mathbb{R}^{3\times 3}{\big |}\bm{R}^{T}\bm{R}=\bm{I}_3, {\rm det}\bm{R}=1\}$. Thus, the configuration of the fixed-wing UAV can be given by
\begin{equation}\label{eq_g_def}
\bm{g}=\begin{bmatrix} \bm{R} & \bm{p} \\ \bm{0}_{1\times 3} & 1 \end{bmatrix}\in\mathbb{R}^{4\times 4}.
\end{equation}
Since $\bm{g}$ is uniquely defined by $\bm{R}$ and $\bm{p}$, it can also be written as $\bm{g}=(\bm{R},\bm{p})$. All the configurations $\bm{g}$ constitute a Lie group named the Special Euclidean group
\begin{equation*}
  {\rm SE(3)}=\left\{\begin{bmatrix} \bm{R} & \bm{p} \\ \bm{0}_{1\times 3} & 1 \end{bmatrix}\in\mathbb{R}^{4\times 4}{\bigg |}\ \bm{R}\in{\rm SO(3)},\ \bm{p}\in\mathbb{R}^3\right\}.
\end{equation*}
Therefore, the configuration $\bm{g}$ of the fixed-wing UAV is an element in the Lie group $\rm{SE(3)}$.

The fixed-wing UAV's angular velocity and linear velocity are denoted by $\bm{\omega}=[\omega^x\ \ \omega^y\ \ \omega^z]^{T}\in\mathbb{R}^3$ and $\bm{v}=[v^x\ \ v^y\ \ v^z]^{T}\in\mathbb{R}^3$, which are both defined in $\bm{\mathcal{F}}_{\rm b}$. Regarding SO(3), the associated Lie algebra is $\mathfrak{so}(3)=\{\bm{S}\in\mathbb{R}^{3\times 3}{\big |}\ \bm{S}^{T} = -\bm{S}\}$. Define a map $\cdot^{\wedge}:\mathbb{R}^3 \to \mathfrak{so}(3)$ by $(\bm{a}^{\wedge})\bm{b}=\bm{a}\times \bm{b}$ for all $\bm{a},\bm{b}\in\mathbb{R}^3$, where ``$\times$" is the vector cross-product, and let $\cdot^{\vee}: \mathfrak{so}(3) \to \mathbb{R}^3$ denote the inverse isomorphism. Then, the associated Lie algebra of $\rm{SE(3)}$ is
\begin{equation*}
  \mathfrak{se}(3)=\left\{\begin{bmatrix}
                       \bm{S} & \bm{v} \\
                       \bm{0}_{1\times 3} & 0
                     \end{bmatrix}\in\mathbb{R}^{4\times 4}{\bigg |}\ \bm{S}\in\mathfrak{so}(3), \bm{v}\in\mathbb{R}^3 \right\}.
\end{equation*}
With slight abuse of notation, we introduce the linear map $\cdot^{\wedge}:\mathbb{R}^3\oplus\mathbb{R}^3 \to \mathfrak{se}(3)$ defined by
\begin{equation}\label{eq_xi_def}
  \bm{\xi}^{\wedge}=\begin{bmatrix}
              \bm{\omega}^{\wedge} & \bm{v} \\
              \bm{0}_{1\times 3} & 0
            \end{bmatrix}\in\mathbb{R}^{4\times 4},
\end{equation}
for $\bm{\xi}=(\bm{\omega},\bm{v})\in\mathbb{R}^3\oplus\mathbb{R}^3$. Similarly, let $\cdot^{\vee}: \mathfrak{se}(3) \to \mathbb{R}^3\oplus\mathbb{R}^3$ represent the inverse isomorphism. The elements in $\mathfrak{se}(3)$ are referred to as twists. Thus, the velocity $\bm{\omega}\in\mathbb{R}^3$ and $\bm{v}\in\mathbb{R}^3$ are formulated to be the twist $\bm{\xi}^{\wedge}\in\mathfrak{se}(3)$.

Having defined the configuration $\bm{g}$ and velocity $\bm{\xi}^{\wedge}$, we can establish the kinematics of the fixed-wing UAV by
\begin{equation}\label{eq_kine}
  \dot{\bm{g}}=\bm{g}\bm{\xi}^{\wedge},
\end{equation}
where $\bm{g}$ is the state and $\bm{\xi}^{\wedge}$ is the control input. It is observed that the system (\ref{eq_kine}) is a global description independent of local coordinates. In addition to the kinematics, the adjoint map ${\rm Ad}_{\bm{g}}:\mathfrak{se}(3)\to\mathfrak{se}(3)$ will be used in later derivation, which is defined by
\begin{equation}\label{eq_Ad_def}
  {\rm Ad}_{\bm{g}}\bm{\eta}^{\wedge}=\bm{g}(\bm{\eta}^{\wedge})\bm{g}^{-1},\quad \forall\bm{\eta}\in\mathfrak{se}(3).
\end{equation}

Although having six DOFs, the fixed-wing UAV is restricted by the nonholonomic constraints, where the word ``nonholonomic" refers to the nonintegrability of a differential equation \cite{Bloch2003Nonholonomic}. Intuitively, nonholonomic constraints of a fixed-wing UAV exhibit as zero linear velocities along the $y-,z-$axis of $\bm{\mathcal{F}}_{\rm b}$, i.e., $v^y=v^z=0$, which can be expressed in a coordinate-free formulation as
\begin{equation}\label{eq_e_xi_e_0}
  \begin{bmatrix}
    \bm{e}_2 & \bm{e}_3
  \end{bmatrix}^T(\bm{\xi}^{\wedge})\bm{e}_4=\bm{0}_{2\times 1},
\end{equation}
where $\bm{e}_i\in\mathbb{R}^4$ $(i=2,3,4)$ are the unit vectors with the $i$-th entry as $1$. Hence, the nonholonomic constraints render the fixed-wing UAV to be an underactuated system in the sense that the six DOFs of the UAV are controlled only by four available control inputs, i.e., $v^x$, $\omega^x$, $\omega^y$, $\omega^z$.

Thus, the kinematics of the fixed-wing UAV is described by (\ref{eq_kine}) together with the nonholonomic constraints (\ref{eq_e_xi_e_0}). Although it seems abstract, the kinematics can be transformed to a common formulation by using the Euler angles. Let $\phi,\theta,\psi$ denote the roll, pitch, yaw angles, respectively, and we employ the rotation sequence of $\phi-\theta-\psi$. Then, the rotation matrix $\bm{R}$ can be parameterized to be
\begin{equation}\label{eq_R_angles}
  \bm{R}=\begin{bmatrix}
           {\rm c}_{\psi}{\rm c}_{\theta} & -{\rm s}_{\psi}{\rm c}_{\phi}+{\rm c}_{\psi}{\rm s}_{\theta}{\rm s}_{\phi} & {\rm s}_{\psi}{\rm s}_{\phi}+{\rm c}_{\psi}{\rm s}_{\theta}{\rm c}_{\phi} \\
           {\rm s}_{\psi}{\rm c}_{\theta} & {\rm c}_{\psi}{\rm c}_{\phi}+{\rm s}_{\psi}{\rm s}_{\theta}{\rm s}_{\phi} & -{\rm c}_{\psi}{\rm s}_{\phi}+{\rm s}_{\psi}{\rm s}_{\theta}{\rm c}_{\phi} \\
           -{\rm s}_{\theta} & {\rm c}_{\theta}{\rm s}_{\phi} & {\rm c}_{\theta}{\rm c}_{\phi}
         \end{bmatrix},
\end{equation}
where ${\rm c}_{\star}\triangleq\cos\star$ and ${\rm s}_{\star}\triangleq\sin\star$. We take the time derivative of $\bm{R}$, and the rotation kinematics in the coordinates of the Euler angles can be obtained as
\begin{equation}\label{eq_rot_kin_local}
  \begin{bmatrix}
    \dot{\phi} \\
    \dot{\theta} \\
    \dot{\psi}
  \end{bmatrix}=
  \begin{bmatrix}
    1 & \tan\theta\sin\phi & \tan\theta\cos\phi \\
    0 & \cos\phi & -\sin\phi \\
    0 & \sec\theta\sin\phi & \sec\theta\cos\phi
  \end{bmatrix}
  \begin{bmatrix}
    \omega^x \\
    \omega^y \\
    \omega^z
  \end{bmatrix},
\end{equation}
implying that the rotation kinematics in 3D is not an integrator like in 2D anymore \cite{Beard2012Small}. Regarding the translation kinematics, it follows from (\ref{eq_kine}) that the time derivative of $\bm{p}$ is $\dot{\bm{p}}=\bm{R}\bm{v}$, and by substituting (\ref{eq_R_angles}) into it, we have
\begin{equation}\label{eq_tran_kin_local}
  \begin{bmatrix}
    \dot{x} \\
    \dot{y} \\
    \dot{z}
  \end{bmatrix}=v^x
  \begin{bmatrix}
    \cos\psi\cos\theta \\
    \sin\psi\cos\theta \\
    -\sin\theta
  \end{bmatrix},
\end{equation}
where $v^y=v^z=0$ is utilized. Therefore, the globally-described kinematics (\ref{eq_kine}) can be locally parameterized as the rotation kinematics (\ref{eq_rot_kin_local}) and translation kinematics (\ref{eq_tran_kin_local}). Although providing an intuitive way to represent the attitude, the Euler angles suffer from singularity. Hence, we establish the model in the Lie group ${\rm SE(3)}$ without local coordinates, and employ the kinematics (\ref{eq_kine}) with the nonholonomic constraints (\ref{eq_e_xi_e_0}) to describe the fixed-wing UAV.

\subsection{Basic Graph Theory}

A group of fixed-wing UAVs interact with each other via communication networks, and it is convenient to model the information exchanges by directed graphs. A directed graph $\mathcal{G}$ is a pair $(\mathcal{V},\mathcal{E})$, where $\mathcal{V}=\{\nu_1,\cdots,\nu_N\}$ is a nonempty finite node set and $\mathcal{E}\subseteq\mathcal{V}\times\mathcal{V}$ is an edge set of ordered pairs of nodes, called edges. The edge $(\nu_i,\nu_j)$ in the edge set $\mathcal{E}$ denotes that the node $\nu_j$ can obtain information from the node $\nu_i$, but not necessarily vice versa. For an edge $(\nu_i,\nu_j)$, the node $\nu_i$ is called the parent node, and $\nu_j$ is the child node.

A directed path from node $\nu_{i_1}$ to node $\nu_{i_l}$ is a sequence of ordered edges of the form $(\nu_{i_k}, \nu_{i_{k+1}})$, $k=1,\cdots,l-1$. A cycle is a directed path that starts and ends at the same node. A directed acyclic graph is a directed graph without any cycle. A directed spanning tree is a directed graph in which every node has exactly one parent except for one node, called the root node, which has no parent and has directed paths to all other nodes.

\subsection{RTI formation in 3D}\label{sec_RTI}

Consider $N+1$ fixed-wing UAVs communicated by a directed acyclic graph. Let $\bm{g}_i=(\bm{R}_i,\bm{p}_i)\in{\rm SE(3)}$ and $\bm{\xi}_i^{\wedge}=(\bm{\omega}_i^{\wedge},\bm{v}_i)\in\mathfrak{se}(3)$ denote the configuration and velocity of the fixed-wing UAVs $(i=0,1,\cdots,N)$, where $0$ represents the leader (or the root node of the graph), while $1,\cdots,N$ represent the followers.

\begin{definition}[RTI formation]\label{def_RTI}
  Given a group element $\bar{\bm{g}}_{0i}=(\bar{\bm{R}}_{0i},\bar{\bm{p}}_{0i})\in{\rm SE(3)}$ $(i=1,\cdots,N)$, if the follower $i$'s configuration $\bm{g}_i$ satisfies
  \begin{equation}\label{eq_RTI_def}
    \bm{g}_i=\bm{g}_0\bar{\bm{g}}_{0i},
  \end{equation}
  where $\bm{g}_0$ is the leader's configuration, then it is called that the follower $\bm{g}_i$ achieves the desired formation pattern specified by $\bar{\bm{g}}_{0i}$ with respect to the leader $\bm{g}_0$. Furthermore, if the desired rotation matrix $\bar{\bm{R}}_{0i}$ and position vector $\bar{\bm{p}}_{0i}$ are both constants, in other words, $\bar{\bm{g}}_{0i}$ is a constant group element in ${\rm SE(3)}$, then $\bar{\bm{g}}_{0i}$ is referred to as the roto-translation invariant (RTI) formation. If only the desired position vector $\bar{\bm{p}}_{0i}$ is a constant, then $\bar{\bm{g}}_{0i}$ is referred to as the pseudo roto-translation invariant (P-RTI) formation.
\end{definition}

\begin{remark}
  The formation pattern $\bar{\bm{g}}_{0i}$ describes a relative configuration of $\bm{g}_i$ with respect to $\bm{g}_0$, rather than an absolute configuration. The RTI formation requires that the relative position $\bar{\bm{p}}_{0i}$ should be unchanged in the leader's body-fixed frame, and the relative attitude $\bar{\bm{R}}_{0i}$ with respect to the leader should also keep fixed. Regarding the P-RTI formation, only fixed relative position is required, while no formation constraint is introduced to the relative attitude. However, the relative attitude in the P-RTI formation cannot be set arbitrarily. Instead, it is typically decided by the nonholonomic constraints.
\end{remark}

\emph{Problem Statement}: Regarding $N+1$ fixed-wing UAVs connected by a directed acyclic graph, the investigation of the formation problem involves two aspects as given below.
\begin{enumerate}
  \item Propose the RTI formation feasibility and determine what types of formation patterns can be realized by the fixed-wing UAVs under the nonholonomic and input saturation constraints;
  \item Design the formation controller so as to achieve feasible RTI formation patterns for the fixed-wing UAVs.
\end{enumerate}

\section{Communication topology transformation}\label{sec_topolo}

In this section, we propose a transformation for the communication topology, which is a directed acyclic graph. As shown in Definition~\ref{def_RTI}, the formation pattern $\bar{\bm{g}}_{0i}$ is given with respect to the leader $\bm{g}_0$. Nevertheless, in the swarm of the fixed-wing UAVs, an arbitrary follower $\bm{g}_i$ possibly does not interact with the leader $\bm{g}_0$ directly, while $\bm{g}_i$ always has parent nodes $\bm{g}_{P_i}$, i.e., the nodes from which $\bm{g}_i$ receives information in the communication graph. Thus, the formation of $\bm{g}_i$ should be given with respect to its parent nodes $\bm{g}_{P_i}$. For a directed spanning tree, each follower $\bm{g}_i$ has one and only one parent node $\bm{g}_{P_i}$. Then, we can similarly define $\bar{\bm{g}}_{P_ii}$, that is, the formation pattern of $\bm{g}_i$ with respect to $\bm{g}_{P_i}$, by
\begin{equation}\label{eq_def_form_Pi}
  \bm{g}_i=\bm{g}_{P_i}\bar{\bm{g}}_{P_ii}.
\end{equation}
However, regarding a directed acyclic graph, which is considered in this paper, each follower $\bm{g}_i$ probably has more than one parent node. Then, the main challenge is how to define the formation with respect to multiple parent nodes. To solve this problem, we utilize the geometric convex combination in the Lie group ${\rm SE}(3)$ to construct a virtual parent node for each follower.

\begin{definition}[\cite{Peng2019Geometric}, geometric convex combination]\label{def_convex}
  Let $\bm{g}_{P_i}^1,\cdots,\bm{g}_{P_i}^{M_i}$ denote the configurations of the parent nodes of the follower $\bm{g}_i$, where $M_i$ is the number of the parent nodes of $\bm{g}_i$. Then, the geometric convex combination of $\bm{g}_{P_i}^1,\cdots,\bm{g}_{P_i}^{M_i}$, denoted by $\bm{g}_{G_i}$, is iteratively defined by
  \begin{equation}\label{eq_g_convex}
    \begin{aligned}
      \bm{g}_{P_i}^{1,2} &= \bm{g}_{P_i}^1\exp(\lambda_i^1\log((\bm{g}_{P_i}^1)^{-1}\bm{g}_{P_i}^2)), \\
      \bm{g}_{P_i}^{1,2,3} &= \bm{g}_{P_i}^{1,2}\exp(\lambda_i^2\log((\bm{g}_{P_i}^{1,2})^{-1}\bm{g}_{P_i}^3)), \\
      &\cdots \\
      \bm{g}_{G_i} &= \bm{g}_{P_i}^{1,\cdots,M_i-1}\exp(\lambda_i^{M_i-1}\log((\bm{g}_{P_i}^{1,\cdots,M_i-1})^{-1}\bm{g}_{P_i}^{M_i})),
    \end{aligned}
  \end{equation}
  where $\lambda_i^j$ are the convex combination coefficients satisfying $0\le\lambda_i^j\le 1$, $j=1,\cdots,M_i-1$, and the exponential map $\exp:\mathfrak{se}(3)\to{\rm SE}(3)$ and the logarithmic map $\log:{\rm SE}(3)\to\mathfrak{se}(3)$ are given in \cite{Bullo2005Geometric}.
\end{definition}

Regarding an arbitrary follower $\bm{g}_i$, Definition~\ref{def_convex} provides a virtual parent node $\bm{g}_{G_i}$, which is defined by the geometric convex combination of the parent nodes of $\bm{g}_i$. Since the graph is acyclic, the geometric convex combination $\bm{g}_{G_i}$ does not involve the information of $\bm{g}_i$, so that there is no algebraic loop regarding $\bm{g}_i$. Similar to every node in the graph, the geometric convex combination $\bm{g}_{G_i}$ also possesses its own kinematics, which is given in the following lemma.

\begin{lemma}\label{lem_kine_convex}
  Let $\bm{\xi}_{P_i}^1,\cdots,\bm{\xi}_{P_i}^{M_i}$ denote the velocities of the parent nodes of $\bm{g}_i$. Then, the kinematics of the geometric convex combination $\bm{g}_{G_i}$ is given by
  \begin{equation}\label{eq_kin_g_G}
    \dot{\bm{g}}_{G_i}=\bm{g}_{G_i}\bm{\xi}_{G_i}^{\wedge},
  \end{equation}
  where $\bm{\xi}_{G_i}$ is the convex combination of the velocities and is defined by
  \begin{equation}\label{eq_xi_convex}
    \begin{aligned}
      \bm{\xi}_{P_i}^{1,2} &= (1-\lambda_i^1)\bm{\xi}_{P_i}^1+\lambda_i^1\bm{\xi}_{P_i}^2, \\
      \bm{\xi}_{P_i}^{1,2,3} &= (1-\lambda_i^2)\bm{\xi}_{P_i}^{1,2}+\lambda_i^2\bm{\xi}_{P_i}^3, \\
      &\cdots \\
      \bm{\xi}_{G_i} &= (1-\lambda_i^{M_i-1})\bm{\xi}_{P_i}^{1,\cdots,M_i-1}+\lambda_i^{M_i-1}\bm{\xi}_{P_i}^{M_i}.
    \end{aligned}
  \end{equation}
  Furthermore, if the velocities $\bm{\xi}_{P_i}^1,\cdots,\bm{\xi}_{P_i}^{M_i}$ are all restricted by the nonholonomic constraints (\ref{eq_e_xi_e_0}), then the velocity convex combination $\bm{\xi}_{G_i}$ defined in (\ref{eq_xi_convex}) satisfies the same nonholonomic constraints
  \begin{equation}\label{eq_xi_convex_nonholo_constrain}
    \begin{bmatrix}
      \bm{e}_2 & \bm{e}_3
    \end{bmatrix}^T(\bm{\xi}_{G_i}^{\wedge})\bm{e}_4=\bm{0}_{2\times 1}.
  \end{equation}
\end{lemma}

\begin{proof}
  It has been proved in \cite{Peng2019Geometric} that the geometric convex combination $\bm{g}_{G_i}$ has the kinematics as (\ref{eq_kin_g_G}). In the following, we prove that the velocity convex combination $\bm{\xi}_{G_i}$ satisfies the nonholonomic constraints. According to (\ref{eq_xi_convex}), it is obtained that
  \begin{align}
    &\begin{bmatrix}
      \bm{e}_2 & \bm{e}_3
    \end{bmatrix}^T(\bm{\xi}_{P_i}^{1,2})^{\wedge}\bm{e}_4 \nonumber \\
    =
    &\begin{bmatrix}
      \bm{e}_2 & \bm{e}_3
    \end{bmatrix}^T((1-\lambda_i^1)(\bm{\xi}_{P_i}^1)^{\wedge}+\lambda_i^1(\bm{\xi}_{P_i}^2)^{\wedge})\bm{e}_4 \nonumber \\
    =
    &(1-\lambda_i^1)
    \begin{bmatrix}
      \bm{e}_2 & \bm{e}_3
    \end{bmatrix}^T((\bm{\xi}_{P_i}^1)^{\wedge})\bm{e}_4 \nonumber \\
    &+
    \lambda_i^1
    \begin{bmatrix}
      \bm{e}_2 & \bm{e}_3
    \end{bmatrix}^T((\bm{\xi}_{P_i}^2)^{\wedge})\bm{e}_4. \label{eq_convex_nonhol_pf}
  \end{align}
  Since $[\bm{e}_2\ \ \bm{e}_3]^T((\bm{\xi}_{P_i}^1)^\wedge)\bm{e}_4=\bm{0}$ and $[\bm{e}_2\ \ \bm{e}_3]^T((\bm{\xi}_{P_i}^2)^\wedge)\bm{e}_4=\bm{0}$, by substituting them into (\ref{eq_convex_nonhol_pf}), it is further derived that
  \begin{equation}
    \begin{bmatrix}
      \bm{e}_2 & \bm{e}_3
    \end{bmatrix}^T(\bm{\xi}_{P_i}^{1,2})^{\wedge}\bm{e}_4=\bm{0}.
  \end{equation}
  Then, we can obtain (\ref{eq_xi_convex_nonholo_constrain}) by iterative computation.
\end{proof}

Therefore, based on Lemma~\ref{lem_kine_convex}, the virtual parent node $\bm{g}_{G_i}$ is described by the kinematics (\ref{eq_kin_g_G}) with the nonholonomic constraints (\ref{eq_xi_convex_nonholo_constrain}), which indicates that the virtual parent node has the same model as any of fixed-wing UAV in the swarm. Furthermore, given the fact that there is no cycle in the communication graph, then for each follower $\bm{g}_i$, the motion of its virtual parent node $\bm{g}_{G_i}$ is independent of $\bm{g}_i$'s own motion. Thus, similar to (\ref{eq_def_form_Pi}), we define the formation pattern of $\bm{g}_i$ with respect to $\bm{g}_{G_i}$, denoted by $\bar{\bm{g}}_{G_ii}$, in the following equality
\begin{equation}\label{eq_def_form_Gi}
  \bm{g}_i=\bm{g}_{G_i}\bar{\bm{g}}_{G_ii}.
\end{equation}
This demonstrates that, with the help of the geometric convex combination in the Lie group ${\rm SE(3)}$, the leader-follower formation of the fixed-wing UAVs in a directed acyclic graph can be handled by the formation of the pair of $\bm{g}_i$ and $\bm{g}_{G_i}$ as given in (\ref{eq_def_form_Gi}). Note that such a case is equivalent to the case of a leader with one follower; for simplicity, we let $\bm{g}_L,\bm{g}_F,\bar{\bm{g}}$ denote $\bm{g}_{G_i},\bm{g}_{i},\bar{\bm{g}}_{G_ii}$, respectively. Then, based on (\ref{eq_def_form_Gi}), the follower $\bm{g}_F$ achieves the formation $\bar{\bm{g}}$ with respect to the leader $\bm{g}_L$, if there holds
\begin{equation}\label{eq_RTI_def_L_F}
  \bm{g}_{F}=\bm{g}_L\bar{\bm{g}},
\end{equation}
which is referred to as the condition of formation achievement. Therefore, in the subsequent sections, we shall study the formation feasibility and design the formation controller based on (\ref{eq_RTI_def_L_F}).

\section{RTI formation feasibility analysis} \label{sec_RTI_feasi}

In this section, we investigate the RTI formation feasibility under two kinds of constraints: the nonholonomic constraints and the input saturation constraints. It is obvious that the condition (\ref{eq_RTI_def_L_F}) provides a requirement for the configurations. But in order to maintain such a formation pattern, the velocities should be restricted as well. Thus, we give the condition of formation maintenance as follows, which is firstly proposed in \cite{He2021Roto}.

\begin{lemma}[\cite{He2021Roto}]\label{lem_maintain}
  The follower $\bm{g}_F$ maintains the desired formation pattern $\bar{\bm{g}}$ with respect to the leader $\bm{g}_L$, if the leader's velocity $\bm{\xi}^{\wedge}_L$ and follower's velocity $\bm{\xi}^{\wedge}_F$ satisfy
  \begin{equation}\label{eq_xi1_Adxi_0}
    \bm{\xi}^{\wedge}_F={\rm Ad}_{\bar{\bm{g}}^{-1}}\bm{\xi}^{\wedge}_L,
  \end{equation}
  where ${\rm Ad}_{\bar{\bm{g}}^{-1}}$ is the adjoint map defined in (\ref{eq_Ad_def}).
\end{lemma}

Therefore, the conditions (\ref{eq_RTI_def_L_F}) and (\ref{eq_xi1_Adxi_0}) guarantee that the RTI formation can be achieved and maintained, which propose the requirements for configurations and velocities, respectively. Regarding the fully actuated systems, which do not impose any velocity constraints, the velocity condition (\ref{eq_xi1_Adxi_0}) can always be satisfied for arbitrary formation pattern $\bar{\bm{g}}$. However, the velocity of the fixed-wing UAVs is restricted by the nonholonomic and input saturation constraints, indicating that the velocity condition (\ref{eq_xi1_Adxi_0}) may not hold for certain types of $\bar{\bm{g}}$. Then, we would like to investigate under what conditions of $\bar{\bm{g}}$, the velocity condition (\ref{eq_xi1_Adxi_0}) can be guaranteed, which is referred to as the formation feasibility as defined below.

\begin{definition}[formation feasibility]\label{def_feasible}
  Let $\Pi_E:\mathfrak{se}(3)\to\mathbb{R}$ and $\Pi_I:\mathfrak{se}(3)\to\mathbb{R}$ represent the nonholonomic constraints and saturation constraints, respectively. Assume that the velocity of the fixed-wing UAVs satisfies  $\Pi_E(\bm{\xi}^\wedge)=0$ and $\Pi_I(\bm{\xi}^\wedge)\le 0$. Then, we call the formation pattern $\bar{\bm{g}}$ is feasible for the fixed-wing UAVs, if there holds
  \begin{align}\label{eq_constraint_equali}
    \Pi_E({\rm Ad}_{\bar{\bm{g}}^{-1}}\bm{\xi}^{\wedge}_L) & =0, \\
    \Pi_I({\rm Ad}_{\bar{\bm{g}}^{-1}}\bm{\xi}^{\wedge}_L) & \le 0. \label{eq_constraint_inequali}
  \end{align}
\end{definition}

Thus, in the following, we study the formation feasibility under these two kinds of constraints, which lays the foundation to the subsequent formation control.

\subsection{Feasibility under nonholonomic constraints}

Basically, Definition~\ref{def_feasible} provides the idea of studying the formation feasibility, that is, embeding the nonholonomic constraints into the velocity condition (\ref{eq_xi1_Adxi_0}), and then figuring out what conditions the formation pattern should satisfy. Before the main results, a lemma of nonholonomic adjoint velocity is given below.

\begin{lemma}\label{le_form_feas}
  Let $\bar{\bm{g}}=(\bar{\bm{R}},\bar{\bm{p}})$ denote the desired formation pattern of the follower $\bm{g}_F$ with respect to the leader $\bm{g}_L$. Regarding the leader's velocity $\bm{\xi}_L^{\wedge}=(\bm{\omega}_L^{\wedge},\bm{v}_L)$, the formation pattern $\bar{\bm{g}}$ is feasible for the follower $\bm{g}_F$ under the nonholonomic constraints, if there holds
  \begin{equation}\label{eq_nonholo_ad_orbit}
    \begin{bmatrix}
      \bm{e}_2 & \bm{e}_3
    \end{bmatrix}^T\bar{\bm{R}}^T(\bm{\omega}_L^{\wedge}\bar{\bm{p}}+\bm{v}_L)=\bm{0}_{2\times 1}.
  \end{equation}
\end{lemma}

\begin{proof}
   We firstly compute the adjoint velocity ${\rm Ad}_{\bar{\bm{g}}^{-1}}\bm{\xi}^{\wedge}_L$. According to the adjoint map defined in (\ref{eq_Ad_def}), we have
  \begin{align}\label{eq_Ad_gbar_xi0}
  {\rm Ad}_{\bar{\bm{g}}^{-1}}\bm{\xi}^{\wedge}_L &=
  \begin{bmatrix}
    \bar{\bm{R}}^T & -\bar{\bm{R}}^T\bar{\bm{p}} \\
    \bm{0}_{1\times 3} & 1
  \end{bmatrix}
  \begin{bmatrix}
    \bm{\omega}_L^{\wedge} & \bm{v}_L \\
    \bm{0}_{1\times 3} & 0
  \end{bmatrix}
  \begin{bmatrix}
    \bar{\bm{R}} & \bar{\bm{p}} \\
    \bm{0}_{1\times 3} & 1
  \end{bmatrix} \nonumber \\ &=
  \begin{bmatrix}
    \bar{\bm{R}}^T\bm{\omega}_L^{\wedge}\bar{\bm{R}} & \bar{\bm{R}}^T(\bm{\omega}_L^{\wedge}\bar{\bm{p}}+\bm{v}_L) \\
    \bm{0}_{1\times 3} & 1
  \end{bmatrix}.
  \end{align}
  For the sake of illustration, we define
  \begin{align}
    \bm{\omega}_{\rm Ad}^{\wedge} &= \bar{\bm{R}}^T\bm{\omega}_L^{\wedge}\bar{\bm{R}}, \label{eq_Ad_omega} \\
    \bm{v}_{\rm Ad} &= \bar{\bm{R}}^T(\bm{\omega}_L^{\wedge}\bar{\bm{p}}+\bm{v}_L), \label{eq_Ad_v}
  \end{align}
  where $\bm{\omega}_{\rm Ad}^{\wedge}$ and $\bm{v}_{\rm Ad}$ are regarded as the angular velocity and linear velocity of ${\rm Ad}_{\bar{\bm{g}}^{-1}}\bm{\xi}^{\wedge}_L$, respectively. Note that $\bm{\xi}_F^{\wedge}$ is restricted by the nonholonomic constraints (\ref{eq_e_xi_e_0}). Then, according to Definition~\ref{def_feasible}, ${\rm Ad}_{\bar{\bm{g}}^{-1}}\bm{\xi}^{\wedge}_L$ should satisfy such constraints as well. Thus, we embed the nonholonomic constraints (\ref{eq_e_xi_e_0}) into (\ref{eq_Ad_gbar_xi0}), and it follows that
  \begin{equation}
    \begin{bmatrix}
      \bm{e}_2 & \bm{e}_3
    \end{bmatrix}^T({\rm Ad}_{\bar{\bm{g}}^{-1}}\bm{\xi}^{\wedge}_L)\bm{e}_4=\bm{0}_{2\times 1},
  \end{equation}
  which is further simplified to be (\ref{eq_nonholo_ad_orbit}) based on (\ref{eq_Ad_gbar_xi0}).
\end{proof}

It is observed that the constraints (\ref{eq_nonholo_ad_orbit}) involve the leader's velocity $\bm{\xi}_L^{\wedge}=(\bm{\omega}_L^{\wedge},\bm{v}_L)$ and the formation pattern $\bar{\bm{g}}=(\bar{\bm{R}},\bar{\bm{p}})$. Generally, the leader's velocity has been defined in advance and independent of the formation controller. In contrast, the formation pattern $\bar{\bm{g}}$ is specified by users according to the task. Then, the constraints (\ref{eq_nonholo_ad_orbit}) rely on the choice of $(\bar{\bm{R}},\bar{\bm{p}})$. Moreover, the position vector $\bar{\bm{p}}$ in the formation can usually be specified arbitrarily, while the rotation matrix $\bar{\bm{R}}$ cannot. This is because the nonholonomic constraints restrict the direction of the linear velocity, i.e., the orientation of the fixed-wing UAV. Hence, the rotation matrix $\bar{\bm{R}}$ in the formation pattern should possess a particular form to satisfy the nonholonomic constraints, which exhibits as (\ref{eq_nonholo_ad_orbit}). Then, it is implied that we can propose the formation feasibility by investigating what kinds of $\bar{\bm{R}}$ make the constraints (\ref{eq_nonholo_ad_orbit}) hold.

It is not difficult to find out that the rotation matrix $\bar{\bm{R}}$ could be obtained by solving the matrix equation (\ref{eq_nonholo_ad_orbit}), where the variables except for $\bar{\bm{R}}$ are all known. However, it is not trivial to solve such a matrix equation directly. Moreover, when defining an attitude in applications, we typically provide the Euler angles and rotation sequence firstly, owing to the explicit physical interpretations, and then derive the rotation matrix by the transformation in (\ref{eq_R_angles}). Thus, it will be more convenient to define the formation pattern in real-word scenarios, if the formation feasibility under the nonholonomic constraints is characterized from the perspective of Euler angles. Therefore, we reformulate the (\ref{eq_nonholo_ad_orbit}) by using Euler angles, and propose the conditions when such constraints are able to hold.

\begin{figure*}
\normalsize
\setcounter{MYtempeqncnt}{\value{equation}}
\setcounter{equation}{29}
\begin{equation}
\label{eq_dbl_1}
\left(-\sin\bar{\psi}\cos\bar{\phi}+\cos\bar{\psi}\sin\bar{\theta}\sin\bar{\phi}\right)\tau_{\rm Ad}^x +
   \left(\cos\bar{\psi}\cos\bar{\phi}+\sin\bar{\psi}\sin\bar{\theta}\sin\bar{\phi}\right)\tau_{\rm Ad}^y + \cos\bar{\theta}\sin\bar{\phi}\tau_{\rm Ad}^z =0
\end{equation}
\begin{equation}
\label{eq_dbl_2}
  \left(\sin\bar{\psi}\sin\bar{\phi}+\cos\bar{\psi}\sin\bar{\theta}\cos\bar{\phi}\right)\tau_{\rm Ad}^x + \left(-\cos\bar{\psi}\sin\bar{\phi}+\sin\bar{\psi}\sin\bar{\theta}\cos\bar{\phi}\right)\tau_{\rm Ad}^y
  + \cos\bar{\theta}\cos\bar{\phi}\tau_{\rm Ad}^z =0.
\end{equation}
\setcounter{equation}{\value{MYtempeqncnt}}
\hrulefill
\end{figure*}

Let $\bar{\phi},\bar{\theta},\bar{\psi}$ denote the desired roll, pitch, yaw angles of the follower $\bm{g}_F$ with respect to the leader $\bm{g}_L$. Then, the formation feasibility is presented in the following theorem.

\begin{theorem}\label{theo_RTI_general}
  Consider a formation pattern $\bar{\bm{g}}=(\bar{\bm{R}},\bar{\bm{p}})$, which describes the desired relative configuration of the follower $\bm{g}_F$ with respect to the leader $\bm{g}_L$. Let $\bar{\bm{p}}=[\bar{x}\ \ \bar{y}\ \ \bar{z}]^{T}$ be an arbitrary constant vector in $\mathbb{R}^3$, and let $\bar{\bm{R}}$ be decided by a set of Euler angles $(\bar{\phi},\bar{\theta},\bar{\psi})$, where $\bar{\phi}$ is arbitrarily chosen in $(-\pi,\pi)$, $\bar{\theta}$ is given by
  \begin{equation}\label{eq_theta_feas}
    \bar{\theta}=-\arctan\frac{\tau_{\rm Ad}^z}{\cos\bar{\psi}\tau_{\rm Ad}^x+\sin\bar{\psi}\tau_{\rm Ad}^y},
  \end{equation}
  and $\bar{\psi}$ is given by
  \begin{equation}\label{eq_psi_feas}
    \bar{\psi}=\arctan\frac{\tau_{\rm Ad}^y}{\tau_{\rm Ad}^x}.
  \end{equation}
  Herein, $\tau_{\rm Ad}^x,\tau_{\rm Ad}^y,\tau_{\rm Ad}^z$ are the components of $\bm{\omega}_L^{\wedge}\bar{\bm{p}}+\bm{v}_L$. Then, the following statements hold.
  \begin{enumerate}[1)]
    \item The formation pattern $\bar{\bm{g}}$ is feasible for the fixed-wing UAVs under the nonholonomic constraints.
    \item The formation pattern $\bar{\bm{g}}$ is RTI, if the leader's velocity $\bm{\xi}_L^{\wedge} =(\bm{\omega}_L^{\wedge},\bm{v}_L)$ is time-invariant.
    \item The formation pattern $\bar{\bm{g}}$ is P-RTI, if the leader's velocity $\bm{\xi}_L^{\wedge} =(\bm{\omega}_L^{\wedge},\bm{v}_L)$ is time-varying.
  \end{enumerate}
\end{theorem}

\begin{proof}
  1) Since the relative position $\bar{\bm{p}}$ has been predefined as constant, then based on Lemma~\ref{le_form_feas}, the feasibility of the formation patten $\bar{\bm{g}}$ is equivalent to making (\ref{eq_nonholo_ad_orbit}) hold by choosing appropriate $\bar{\bm{R}}$.
  Note that if $\bm{\omega}_L^{\wedge}\bar{\bm{p}}+\bm{v}_L=\bm{0}$, the rotation matrix $\bar{\bm{R}}$ can be arbitrarily chosen in ${\rm SO(3)}$. However, such a case results in $\bm{v}_{\rm Ad}=\bm{0}$ based on (\ref{eq_Ad_v}). According to (\ref{eq_xi1_Adxi_0}), this further leads to the follower's linear velocity satisfying $\bm{v}_F=\bm{v}_{\rm Ad}=\bm{0}$. Nevertheless, the fixed-wing UAVs cannot hover in the space, in other words, the linear velocity cannot be zero. Thus, the case of $\bm{\omega}_L^{\wedge}\bar{\bm{p}}+\bm{v}_L=\bm{0}$ is not considered in this paper. Regarding $\bm{\omega}_L^{\wedge}\bar{\bm{p}}+\bm{v}_L\ne\bm{0}$, it is obtained that
  \begin{equation}\label{eq_tau123}
    \bm{\omega}_L^{\wedge}\bar{\bm{p}}+\bm{v}_L=
    \begin{bmatrix}
      \omega_L^y\bar{p}^z-\omega_L^z\bar{p}^y+v_L^x \\
      \omega_L^z\bar{p}^x-\omega_L^x\bar{p}^z \\
      \omega_L^x\bar{p}^y-\omega_L^y\bar{p}^x
    \end{bmatrix}\triangleq
    \begin{bmatrix}
      \tau_{\rm Ad}^x \\
      \tau_{\rm Ad}^y \\
      \tau_{\rm Ad}^z
    \end{bmatrix}.
  \end{equation}
  Additionally, the rotation matrix $\bar{\bm{R}}$ can be computed base on $\bar{\phi},\bar{\theta},\bar{\psi}$ according to the transformation in (\ref{eq_R_angles}). We substitute the $\bar{\bm{R}}$ in (\ref{eq_R_angles}) and the $\bm{\omega}_L^{\wedge}\bar{\bm{p}}+\bm{v}_L$ in (\ref{eq_tau123}) into (\ref{eq_nonholo_ad_orbit}), and after a series of derivation, the constraints (\ref{eq_nonholo_ad_orbit}) can be reformulated as two algebraic equations (\ref{eq_dbl_1}) and (\ref{eq_dbl_2}). Thus, according to Lemma~\ref{le_form_feas}, the formation pattern $\bar{\bm{g}}$ is feasible if these two equations (\ref{eq_dbl_1}) and (\ref{eq_dbl_2}) hold. Since $\bar{\bm{p}}$ and $\bm{\xi}_L^{\wedge}$ are both predefined, the variables $\tau_{\rm Ad}^x,\tau_{\rm Ad}^y,\tau_{\rm Ad}^z$ are all known. Therefore, we can propose conditions to ensure formation feasibility by deriving the Euler angles $\bar{\phi},\bar{\theta},\bar{\psi}$ from (\ref{eq_dbl_1}) and (\ref{eq_dbl_2}) straightforwardly.

  We firstly derive the yaw angle $\bar{\psi}$. By the transformation of (\ref{eq_dbl_1})$\times\cos\bar{\phi}+$(\ref{eq_dbl_2})$\times\sin\bar{\phi}$, we cancel the terms involving $\bar{\theta}$ in (\ref{eq_dbl_1}) and (\ref{eq_dbl_2}), and obtain that
  \setcounter{equation}{31} 
  \begin{equation}
    (\sin^2\bar{\phi}+\cos^2\bar{\phi})\sin\bar{\psi}\tau_{\rm Ad}^x-(\sin^2\bar{\phi}+\cos^2\bar{\phi})\cos\bar{\psi}\tau_{\rm Ad}^y=0.
  \end{equation}
  Due to $\sin^2\bar{\phi}+\cos^2\bar{\phi}=1$, we have
  \begin{equation}\label{eq_psi_feas_ex}
    \sin\bar{\psi}\tau_{\rm Ad}^x=\cos\bar{\psi}\tau_{\rm Ad}^y,
  \end{equation}
  which is further written to be
  \begin{equation}\label{eq_psi_feas_ex_2}
    \tan\bar{\psi}=\frac{\tau_{\rm Ad}^y}{\tau_{\rm Ad}^x}.
  \end{equation}
  So that, the yaw angle $\bar{\psi}$ is obtained as in (\ref{eq_psi_feas}). Next, the pitch angle $\bar{\theta}$ will be derived. By the transformation of (\ref{eq_dbl_1})$\times\sin\bar{\phi}+$(\ref{eq_dbl_2})$\times\cos\bar{\phi}$, the terms involving $\bar{\phi}$ are cancelled and it follows that
  \begin{equation}\label{eq_theta_feas_ex}
    \sin\bar{\theta}(\cos\bar{\psi}\tau_{\rm Ad}^x+\sin\bar{\psi}\tau_{\rm Ad}^y)+\cos\bar{\theta}\tau_{\rm Ad}^z=0.
  \end{equation}
  Then, it is further obtained that
  \begin{equation}\label{eq_theta_feas_ex_2}
    \tan\bar{\theta}=-\frac{\tau_{\rm Ad}^z}{\cos\bar{\psi}\tau_{\rm Ad}^x+\sin\bar{\psi}\tau_{\rm Ad}^y},
  \end{equation}
  where $\bar{\psi}$ is provided in (\ref{eq_psi_feas}). Hence, the pitch angle $\bar{\theta}$ is given by (\ref{eq_theta_feas}). Lastly, the expression of the roll angle $\bar{\phi}$ should be derived. To this end, the formulas (\ref{eq_dbl_1}) and (\ref{eq_dbl_2}) are reorganized to be
  \begin{align}
    & \left(-\sin\bar{\psi}\tau_{\rm Ad}^x+\cos\bar{\psi}\tau_{\rm Ad}^y\right)\cos\bar{\phi} + \sin\bar{\phi}\times \nonumber \\
    & \qquad \left(\sin\bar{\theta}\left(\cos\bar{\psi}\tau_{\rm Ad}^x+\sin\bar{\psi}\tau_{\rm Ad}^y\right) +\cos\bar{\theta}\tau_{\rm Ad}^z\right)=0, \label{eq_phi_feas_ex1}\\
    & \left(\sin\bar{\psi}\tau_{\rm Ad}^x-\cos\bar{\psi}\tau_{\rm Ad}^y\right)\sin\bar{\phi} + \cos\bar{\phi}\times \nonumber \\
    & \qquad \left(\sin\bar{\theta}\left(\cos\bar{\psi}\tau_{\rm Ad}^x+\sin\bar{\psi}\tau_{\rm Ad}^y\right) +\cos\bar{\theta}\tau_{\rm Ad}^z\right)=0. \label{eq_phi_feas_ex2}
  \end{align}
  Substituting (\ref{eq_theta_feas_ex}) into (\ref{eq_phi_feas_ex1}) and (\ref{eq_phi_feas_ex2}), it follows that
  \begin{align}
     \left(\sin\bar{\psi}\tau_{\rm Ad}^x-\cos\bar{\psi}\tau_{\rm Ad}^y\right)\cos\bar{\phi} =0, \label{eq_phi_feas_ex3}\\
     \left(\sin\bar{\psi}\tau_{\rm Ad}^x-\cos\bar{\psi}\tau_{\rm Ad}^y\right)\sin\bar{\phi} =0. \label{eq_phi_feas_ex4}
  \end{align}
  According to (\ref{eq_psi_feas_ex}), we have $\sin\bar{\psi}\tau_{\rm Ad}^x-\cos\bar{\psi}\tau_{\rm Ad}^y=0$. Thus, the formulas (\ref{eq_phi_feas_ex3}) and (\ref{eq_phi_feas_ex4}) always hold no matter what value of the roll angle $\bar{\phi}$ is predefined. Therefore, the formation pattern $\bar{\bm{g}}$ given by (\ref{eq_theta_feas}) and (\ref{eq_psi_feas}) satisfies the constraints (\ref{eq_dbl_1}) and (\ref{eq_dbl_2}), indicating that $\bar{\bm{g}}$ is feasible for the formation of the fixed-wing UAVs.

  2) According to Definition~\ref{def_RTI}, the formation pattern $\bar{\bm{g}}$ is RTI if the rotation matrix $\bar{\bm{R}}$ and position vector $\bar{\bm{p}}$ are both time-invariant. Note that $\bar{\bm{p}}$ has been specified as a constant vector in $\mathbb{R}^3$. Thus, we only have to prove $\bar{\bm{R}}$ is constant. Since the leader's velocity $\bm{\xi}_L^{\wedge}=(\bm{\omega}_L^{\wedge},\bm{v}_L)$ and the position vector $\bar{\bm{p}}$ are both time-invariant, it follows from (\ref{eq_tau123}) that $\tau_{\rm Ad}^x,\tau_{\rm Ad}^y,\tau_{\rm Ad}^z$ are all constants. Then, the attitude angles $\bar{\theta}$ and $\bar{\psi}$ given in (\ref{eq_theta_feas}) and (\ref{eq_psi_feas}) can both keep fixed values. Plus a constant $\bar{\phi}$, the rotation matrix $\bar{\bm{R}}$ given by (\ref{eq_R_angles}) is time-invariant, which indicates that the formation pattern $\bar{\bm{g}}$ is RTI.

  3) Considering the fact that the leader's velocity $\bm{\xi}_L^{\wedge}=(\bm{\omega}_L^{\wedge},\bm{v}_L)$ is time-varying, we can obtain the attitude angles $\bar{\theta}$ and $\bar{\psi}$ given in (\ref{eq_theta_feas}) and (\ref{eq_psi_feas}) will vary with time as well, which contributes to a time-varying rotation matrix $\bar{\bm{R}}$ accordingly. Therefore, only the position vector $\bm{p}$ is constant, implying that the formation pattern $\bar{\bm{g}}$ is P-RTI according to Definition~\ref{def_RTI}.
\end{proof}

\begin{remark}
  The formulas (\ref{eq_dbl_1}) and (\ref{eq_dbl_2}) are regarded as the nonholonomic constraints for the RTI formation. Note that a fixed-wing UAV is nonholonomic constrained in two directions of the translational motion, i.e., swaying and heaving, while it is interesting to observe that these two nonholonomic constraints are eventually formulated to be two constraints (\ref{eq_dbl_1}) and (\ref{eq_dbl_2}) for three attitude angles $\bar{\phi},\bar{\theta},\bar{\psi}$. This further implies an unconstrained rotation DOF is indeed preserved in the choice of the formation pattern, which is verified by the proof of Theorem~\ref{theo_RTI_general} that $\bar{\psi}$ and $\bar{\theta}$ should be given by (\ref{eq_psi_feas}) and (\ref{eq_theta_feas}) to make the RTI formation feasible, while $\bar{\phi}$ is freely chosen as needed.
\end{remark}

In Theorem~\ref{theo_RTI_general}, the desired relative attitude matrix $\bar{\bm{R}}$ is decided by three attitude angles $\bar{\phi},\bar{\theta},\bar{\psi}$. Then, a naturally-arising question is when $\bar{\phi},\bar{\theta},\bar{\psi}$ can all be $0$, followed by the fact that $\bar{\bm{R}}$ degenerates to be the identity matrix $\bm{I}_3$. Intuitively, this case means the follower $\bm{g}_F$ has the same attitude as the leader $\bm{g}_L$. In other words, they point to the same direction in the formation pattern, which is relatively common in real-world scenarios. Therefore, the following theorem illustrates the feasibility of such a particular formation pattern for the fixed-wing UAVs.

\begin{corollary}\label{theo_RTI_identity}
  Let $\bar{\bm{g}}=(\bm{I}_3,\bar{\bm{p}})$ denote an RTI formation pattern, where $\bm{I}_3$ is the identity matrix in $\mathbb{R}^{3\times 3}$ and $\bar{\bm{p}}$ is a constant vector in $\mathbb{R}^3$. Then, $\bar{\bm{g}}$ is feasible for the fixed-wing UAVs under the nonholonomic constraints, if there holds  $\bm{\omega}_L^{\wedge}=\bm{0}_{3\times 3}$ or $\omega_L^x=0,\bar{p}^x=0$.
\end{corollary}

\begin{proof}
   According to Lemma~\ref{le_form_feas}, a formation pattern $\bar{\bm{R}}$ is feasible if the condition (\ref{eq_nonholo_ad_orbit}) holds. Once it has $\bm{\omega}_L^{\wedge}=\bm{0}_{3\times 3}$, the condition (\ref{eq_nonholo_ad_orbit}) degenerates to
  \begin{equation}\label{eq_feas_identity}
    [\bm{e}_2\quad\bm{e}_3]^T\bar{\bm{R}}^T\bm{v}_L=\bm{0}_{2\times 1}.
  \end{equation}
  Due to $\bm{v}_L=[v_L^x\quad 0 \quad 0]^T$, then the formula (\ref{eq_feas_identity}) naturally holds for $\bar{\bm{R}}=\bm{I}_3$, which indicates the formation pattern $\bar{\bm{g}}(\bm{I}_3,\bar{\bm{p}})$ is feasible for the fixed-wing UAVs. Regarding the case of $\omega_L^x=0,\bar{p}^x=0$, it follows from (\ref{eq_tau123}) that
  \begin{equation}\label{eq_feas_identity_2}
    \bm{\omega}_L^{\wedge}\bar{\bm{p}}+\bm{v}_L=[\tau_{\rm Ad}^x\quad 0 \quad 0]^T.
  \end{equation}
   By substituting (\ref{eq_feas_identity_2}) and $\bar{\bm{R}}=\bm{I}_3$ into (\ref{eq_nonholo_ad_orbit}), we can verify the rightness of (\ref{eq_nonholo_ad_orbit}), implying the feasibility of the formation pattern $\bar{\bm{g}}=(\bm{I}_3,\bar{\bm{p}})$.
\end{proof}

\begin{figure}[t]
  \centering
  \includegraphics[width=0.25\textwidth]{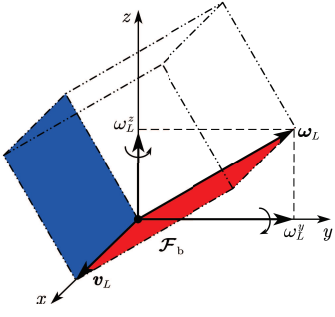}
  \caption{Schematic diagram of the velocities in the case of $\omega_L^x=0$}
  \label{fig_circular}
\end{figure}

\begin{remark}
  In Corollary~\ref{theo_RTI_identity}, it is evident that $\bm{\omega}_L^{\wedge}=\bm{0}_{3\times 3}$ refers to a straight line trajectory, while the case of $\omega_L^x=0$ ($\omega_L^y\ne 0,\omega_L^z\ne 0$) implicitly represents the leader moves along a circle in the plane perpendicular to $\bm{\omega}_L$. As shown in Figure~\ref{fig_circular}, once it has $\omega_L^x=0$, the angular velocity $\bm{\omega}_L$ lies in the $y-z$ plane and is orthogonal to the linear velocity $\bm{v}_L$, followed by the fact that $\bm{\omega}_L$ and $\bm{v}_L$ defines a plane (the red plane). Then, according to theoretical mechanics, the fixed-wing UAVs move along a circle in the blue plane (which is perpendicular to the red plane) with the radius $r=\|\bm{v}_L\|/\|\bm{\omega}_L\|$. Therefore, Corollary~\ref{theo_RTI_identity} demonstrates intuitively that the fixed-wing UAVs have the same orientations in two cases, i.e., moving along a straight line, or a circle with parallel formation patterns.
\end{remark}

\begin{remark}
  It is well known that the main drawback of the Euler angles is the singularity at certain points. However, the formation feasibility and controller proposed in this paper do not suffer from any singularity. The reasons are given below. For the $\phi-\theta-\psi$ sequence, the singularity appears when $\theta=\pm\frac{\pi}{2}$. In such a case, one rotation DOF is lost, so that the rest two angles cannot be well defined. This can be explicitly viewed in mathematics. Once $\theta=\pm\frac{\pi}{2}$, the rotation matrix in (\ref{eq_R_angles}) degenerates to
  \begin{equation}\label{eq_R_angles_degen}
    \bm{R}=\pm\begin{bmatrix}
             0 & \sin(\phi-\psi) & \cos(\phi-\psi) \\
             0 & \cos(\phi-\psi) & -\sin(\phi-\psi) \\
             -1 & 0 & 0
           \end{bmatrix},
  \end{equation}
  which is only decided by the error between $\phi$ and $\psi$. Therefore, given a certain attitude (or a rotation matrix) with $\theta=\pm\frac{\pi}{2}$, we cannot derive the specific $\phi$ and $\psi$ from (\ref{eq_R_angles_degen}) exactly. In short, the singularity arises from the fact that the transformation from rotation matrix to Euler angles is one-to-many. However, in this paper, we use Euler angles $\bar{\phi},\bar{\theta},\bar{\psi}$ to derive rotation matrix $\bar{\bm{R}}$ instead. Once $\bar{\phi},\bar{\theta},\bar{\psi}$ are provided, $\bar{\bm{R}}$ can be uniquely obtained without any singularity according to (\ref{eq_R_angles}). Given the fact that the transformation from the Euler angles to the rotation matrix is a surjective map, $\bar{\bm{R}}$ is always well defined by (\ref{eq_R_angles}), even if in the case of $\theta=\pm\frac{\pi}{2}$. Thus, there does not exist any singularity from the Euler angles.
\end{remark}

\subsection{Feasibility under input saturation constraints}

For the fixed-wing UAVs, the angular velocity and linear velocity are typically restricted by the saturation constraints. Particularly, the linear velocity of the fixed-wing UAV should always be positive and has a lower bound to maintain the flight altitude. Note that the RTI formation requires the followers to rotate with the leader so as to keep the formation pattern invariant. Then, if the leader takes sharp turns during the movement, the followers are required to provide a large velocity to maintain the RTI formation accordingly. However, this probably violates the velocity saturation constraints, leading to the formation pattern unfeasible. Hence, in this subsection, the formation feasibility will be investigated under the input saturation constraint. Before providing the main results, we introduce the following lemma which will be utilized later.

\begin{lemma}[\cite{Bullo2005Geometric}]\label{lem_R_times_a_b}
  Given $\bm{R}\in{\rm SO(3)}$ and the associated adjoint map ${\rm Ad}_{\bm{R}}\in\mathcal{L}(\mathfrak{so}(3);\mathfrak{so}(3))$, where $\mathcal{L}(\mathcal{U};\mathcal{V})$ denote the set of linear maps from the space $\mathcal{U}$ to the space $\mathcal{V}$. Then, the matrix representation of the adjoint map is given by $[{\rm Ad}_{\bm{R}}]=\bm{R}$.
\end{lemma}

Now, we propose the formation feasibility under the input saturation constraint in the following theorem.

\begin{theorem}\label{theo_sat}
  Let $\bar{\bm{g}}=(\bar{\bm{R}},\bar{\bm{p}})$ denote the desired formation pattern. The saturation constraints for the leader's velocity are given by $\|\bm{\omega}_L\|\leq\alpha_L$, $\underline{\beta}_L \leq v_L^x \leq \overline{\beta}_L$, where $\alpha_L>0$, $\overline{\beta}_L>\underline{\beta}_L>0$. Similarly, the saturation constraints for the follower's velocity are given by $\|\bm{\omega}_F\|\leq\alpha_F$, $\underline{\beta}_F \leq v_F^x \leq \overline{\beta}_F$, where $\alpha_F=\alpha_L$, $0<\underline{\beta}_F<\underline{\beta}_L<\overline{\beta}_L<\overline{\beta}_F$. Assume that the leader's velocity has satisfied the saturation constraints. Then, the statements given below illustrate the conditions when the follower's velocity satisfies the saturation constraint from the perspective of the desired formation pattern.
  \begin{enumerate}[1)]
    \item Regarding the angular velocity $\bm{\omega}_F$, there always holds $\|\bm{\omega}_F\|\leq\alpha_F$ for arbitrary desired formation pattern $\bar{\bm{g}}$.
    \item Regarding the linear velocity $v_F^x$, there holds $\underline{\beta}_F \leq v_F^x \leq \overline{\beta}_F$ if the desired relative position $\bar{\bm{p}}$ satisfies
        \begin{equation}\label{eq_p_saturation_constraint}
          \|\bar{\bm{p}}\|=\left\{
          \begin{aligned}
            c_1,\quad {\rm for}\ \|\bm{\omega}_L\|=0, \\
            c_2,\quad {\rm for}\ \|\bm{\omega}_L\|\ne 0,
          \end{aligned}
          \right.
        \end{equation}
        where $c_1$ is an arbitrary positive scalar, and $c_2$ is given by
        \begin{equation}\label{eq_p_saturation_constraint_c2}
          c_2=\min\left\{\frac{\overline{\beta}_F-\overline{\beta}_L}{\alpha_L}, \frac{\underline{\beta}_L-\underline{\beta}_F}{\alpha_L}\right\}.
        \end{equation}
  \end{enumerate}
\end{theorem}

\begin{proof}
  According to Lemma~\ref{lem_maintain} and (\ref{eq_Ad_omega})(\ref{eq_Ad_v}), the formulation of the follower's angular and linear velocities are able to be expressed as
  \begin{align}
    \bm{\omega}_F & = (\bar{\bm{R}}^T\bm{\omega}_L^{\wedge}\bar{\bm{R}})^{\vee}, \label{eq_lemma_sat_cons_omega_F_hat} \\
    \bm{v}_F & =\bar{\bm{R}}^T(\bm{\omega}_L^{\wedge}\bar{\bm{p}}+\bm{v}_L). \label{eq_lemma_sat_cons_v_F}
  \end{align}

  1) The angular velocity $\bm{\omega}_F$ can be rewritten as
  \begin{equation}
    \bm{\omega}_F=({\rm Ad}_{\bar{\bm{R}}^T}\bm{\omega}_L^{\wedge})^{\vee}=[{\rm Ad}_{\bar{\bm{R}}^T}]\bm{\omega}_L.
  \end{equation}
  Then, it follows from Lemma~\ref{lem_R_times_a_b} that
  \begin{equation}\label{eq_lemma_sat_cons_omega_F}
    \bm{\omega}_F=\bar{\bm{R}}^T\bm{\omega}_L,
  \end{equation}
  which becomes a rotation transformation with respect to $\bm{\omega}_L$. Then, with the property of $\bar{\bm{R}}\bar{\bm{R}}^T=\bm{I}_3$, we have
  \begin{equation}
    \|\bm{\omega}_F\|=\sqrt{\bm{\omega}_L^T\bar{\bm{R}}\bar{\bm{R}}^T\bm{\omega}_L}=\sqrt{\bm{\omega}_L^T\bm{\omega}_L}=\|\bm{\omega}_L\| \leq \alpha_L=\alpha_F,
  \end{equation}
  which indicates that $\|\bm{\omega}_F\|\leq\alpha_F$ always holds no matter the choice of $\bar{\bm{g}}$.

  2) Regarding the linear velocity, it can be obtained from (\ref{eq_lemma_sat_cons_v_F}) that
  \begin{equation}
    v_F^x=\| \bm{\omega}_L\bar{\bm{p}}+\bm{v}_L \|.
  \end{equation}
  Due to $\bm{v}_L>\bm{0}$, we further have
  \begin{equation}\label{eq_v_F^x_ineq}
    -\| \bm{\omega}_L\| \|\bar{\bm{p}}\|+ \|\bm{v}_L \| \leq v_F^x \leq \| \bm{\omega}_L\| \|\bar{\bm{p}}\|+ \|\bm{v}_L \|
  \end{equation}
  For the case of $\|\bm{\omega}_L\|=0$, it follows that
  \begin{equation}
    v_F^x=\|\bm{v}_L \|=v_L^x,
  \end{equation}
  which indicates $v_F^x$ is not related to $\|\bar{\bm{p}}\|$ any more. Given that $\underline{\beta}_L \leq v_L^x \leq \overline{\beta}_L$, there naturally holds
  \begin{equation}
    \underline{\beta}_F<\underline{\beta}_L<v_F^x<\overline{\beta}_L<\overline{\beta}_F,
  \end{equation}
  demonstrating the saturation constraint of $v_F^x$ is always satisfied, so that the desired relative position $\bar{\bm{p}}$ can be chosen arbitrarily. For the case of $\|\bm{\omega}_L\|\ne 0$, based on (\ref{eq_v_F^x_ineq}), the constraint $\underline{\beta}_F \leq v_F^x \leq \overline{\beta}_F$ is satisfied if there holds
  \begin{align}
    \| \bm{\omega}_L\| \|\bar{\bm{p}}\|+ \|\bm{v}_L \| & \leq \overline{\beta}_F, \label{eq_lemma_sat_cons_low_1} \\
    -\| \bm{\omega}_L\| \|\bar{\bm{p}}\|+ \|\bm{v}_L \| & \geq \underline{\beta}_F. \label{eq_lemma_sat_cons_upper_1}
  \end{align}
  It can be derived from (\ref{eq_lemma_sat_cons_low_1}) that $\|\bm{p}\|$ should satisfy
  \begin{equation}\label{eq_lemma_sat_cons_low_2}
    \|\bm{p}\|\leq\frac{\overline{\beta}_F-\|\bm{v}_L \|}{\| \bm{\omega}_L\|}.
  \end{equation}
  Note that (\ref{eq_lemma_sat_cons_low_2}) is able to hold if $\|\bm{p}\|$ satisfies
  \begin{equation}\label{eq_lemma_sat_cons_low_3}
    \|\bm{p}\|\leq\frac{\overline{\beta}_F-\|\bm{v}_L \|_{\max}}{\| \bm{\omega}_L\|_{\max}} = \frac{\overline{\beta}_F-\overline{\beta}_L}{\alpha_L},
  \end{equation}
  where $\|\bm{\omega}_L\|\leq\alpha_L$ and $\|\bm{v}_L\|\leq\overline{\beta}_L$ are utilized. Similarly, we can obtain that the inequality (\ref{eq_lemma_sat_cons_upper_1}) holds as long as $\|\bm{p}\|$ satisfies
  \begin{equation}\label{eq_lemma_sat_cons_upper_3}
    \|\bm{p}\|\leq\frac{\|\bm{v}_L\|_{\min}-\underline{\beta}_F}{\|\bm{\omega}_L\|_{\max}} = \frac{\underline{\beta}_L-\underline{\beta}_F}{\alpha_L}.
  \end{equation}
  Based on (\ref{eq_lemma_sat_cons_low_3}) and (\ref{eq_lemma_sat_cons_upper_3}), once $\|\bm{\omega}_L\|\ne 0$, the follower's linear velocity satisfies the constraint $\underline{\beta}_F \leq v_F^x \leq \overline{\beta}_F$ if the desired relative position $\bar{\bm{p}}$ satisfies
  \begin{equation}
    \|\bm{p}\|=\min\left\{\frac{\overline{\beta}_F-\overline{\beta}_L}{\alpha_L}, \frac{\underline{\beta}_L-\underline{\beta}_F}{\alpha_L}\right\}.
  \end{equation}
  In summary, the linear velocity constraint $\underline{\beta}_F \leq v_F^x \leq \overline{\beta}_F$ holds if the condition in (\ref{eq_p_saturation_constraint}) is satisfied.
\end{proof}

\begin{figure}
  \centering
  \includegraphics[width=0.32\textwidth]{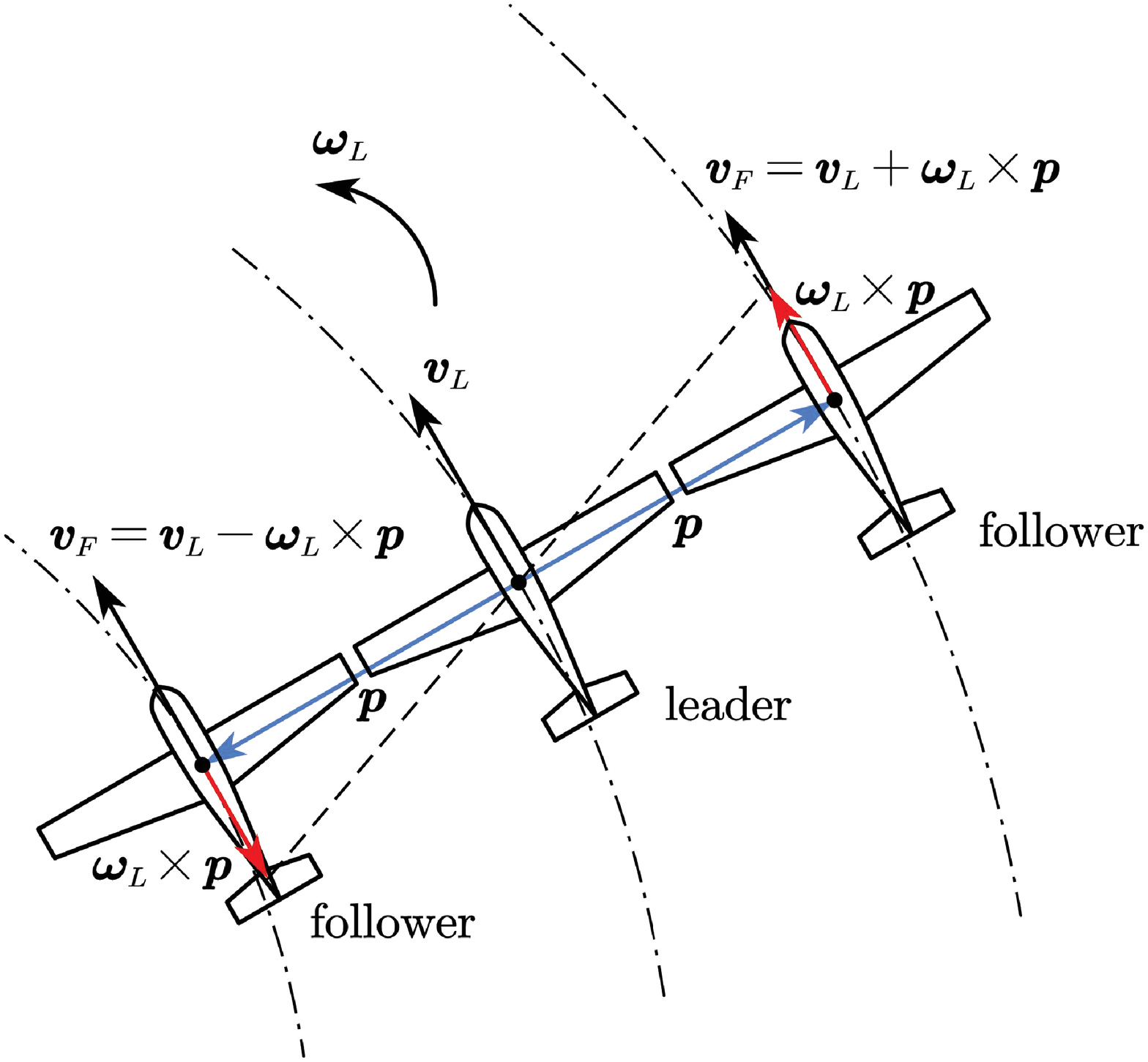}
  \caption{Schematic diagram demonstrating the upper and lower bounds of the follower's saturation constraint.}
  \label{fig_feas_sat}
\end{figure}

\begin{remark}
  It is seen from Theorem~\ref{theo_sat} that the angular speed bounds are identical ($\alpha_F=\alpha_L$), while the linear speed bounds are not ($\underline{\beta}_F<\underline{\beta}_L$, $\overline{\beta}_F>\overline{\beta}_L$). This can be intuitively explained with the aid of Figure~\ref{fig_feas_sat}, where the orientations of all fixed-wing UAVs are identical, and the position vector $\bm{p}$ points laterally. In order to keep RTI, the followers rotate along with the leader at the same angular velocity, i.e., $\bm{\omega}_F=\bm{\omega}_L$. However, regarding the linear velocities, due to the anticlockwise rotation, an additional forward velocity is induced for the right follower, yet an additional backward velocity for the left follower. Then, it is obtained that $\|\bm{v}_F^{\rm outs}\|=\|\bm{v}_L+\bm{\omega}_L\times\bm{p}\|>\|\bm{v}_L\|$ and $\|\bm{v}_F^{\rm ins}\|=\|\bm{v}_L-\bm{\omega}_L\times\bm{p}\|<\|\bm{v}_L\|$. Thus, the linear speed range of the followers should be larger than the leader's.
\end{remark}

\begin{remark}
  Given the fact that the feasibility of $\bar{\bm{p}}$ is always ensured no matter how $\|\bm{\omega}_L\|$ varies in $(0,\alpha_L]$, the conditions of $\|\bar{\bm{p}}\|$ proposed in (\ref{eq_p_saturation_constraint_c2}) may be conservative to some extent. For example, if the angular speed $\|\bm{\omega}_L\|$ is small, that is, $\|\bm{\omega}_L\|=\epsilon$ where $0<\epsilon\ll\alpha_L$, the position norm $\|\bar{\bm{p}}\|$ might not necessarily be chosen as small as (\ref{eq_p_saturation_constraint_c2}). However, in such cases, if it is guaranteed that $\|\bm{\omega}_L\|$ merely varies in a small neighborhood around 0, the value of $\alpha_L$ could be turned down to make the condition (\ref{eq_p_saturation_constraint_c2}) less conservative. Essentially, the range of $\|\bm{\omega}_L\|$ greatly influences the upper bound of $\|\bar{\bm{p}}\|$, so that we can appropriately choose $\alpha_L$ so as to reduce the conservativeness as much as possible.
\end{remark}

\section{RTI formation controller design} \label{sec_control}

In this section, we design the formation controller for the feasible formations of the fixed-wing UAVs. Inspired by \cite{He2021Trajectory}, which deals with the formation tracking of the nonholonomic mobile robots, we introduce a virtual leader related to the desired formation pattern, and then convert leader-follower formation into trajectory tracking. The challenge of designing a formation controller arises from the underactuation of the fixed-wing UAVs, i.e., six DOFs controlled by only four inputs. Therefore, to solve this problem, we design an additional angular velocity to compensate the lack of certain linear velocities. Then, based on the logarithmic feedback in the Lie group ${\rm SE(3)}$, the formation controller is proposed accordingly.

Firstly, we define a virtual leader $\bm{g}_C$, whose configuration is given by
\begin{equation}\label{eq_def_virtual_leader_confi}
  \bm{g}_C=\bm{g}_L\bar{\bm{g}},
\end{equation}
where $\bar{\bm{g}}$ is the desired formation pattern. By taking the time derivative of $\bm{g}_C$, we have
\begin{equation*}
  \dot{\bm{g}}_C=\dot{\bm{g}}_L\bar{\bm{g}}=\bm{g}_L(\bm{\xi}_L^{\wedge})\bar{\bm{g}}=\bm{g}_C{\rm Ad}_{\bar{\bm{g}}^{-1}}\bm{\xi}_L^{\wedge},
\end{equation*}
where the kinematics (\ref{eq_kine}) and the adjoint map (\ref{eq_Ad_def}) are utilized. Define the velocity of the virtual leader as follows
\begin{equation}\label{eq_def_virtual_leader_velocity}
  \bm{\xi}_C^{\wedge}={\rm Ad}_{\bar{\bm{g}}^{-1}}\bm{\xi}_L^{\wedge}.
\end{equation}
Then, the kinematics of the virtual leader is expressed as
\begin{equation}\label{eq_kine_virtual_leader}
  \dot{\bm{g}}_C=\bm{g}_C\bm{\xi}_C^{\wedge}.
\end{equation}

\begin{lemma}\label{lem_form_2_track}
  The follower $\bm{g}_F$ achieves the formation pattern $\bar{\bm{g}}$ with respect to the leader $\bm{g}_L$ if the follower $\bm{g}_F$ tracks the virtual leader $\bm{g}_C$.
\end{lemma}

\begin{proof}
  Once the follower $\bm{g}_F$ tracks the virtual leader $\bm{g}_C$, there holds
  \begin{equation}\label{eq_gF_gC}
    \bm{g}_F=\bm{g}_C.
  \end{equation}
  By substituting (\ref{eq_def_virtual_leader_confi}) into (\ref{eq_gF_gC}), we obtain the condition (\ref{eq_RTI_def_L_F}), which indicates $\bm{g}_F$ achieves the formation pattern $\bar{\bm{g}}$ with respect to $\bm{g}_L$.
\end{proof}

\begin{figure}
  \centering
  \includegraphics[width=0.4\textwidth]{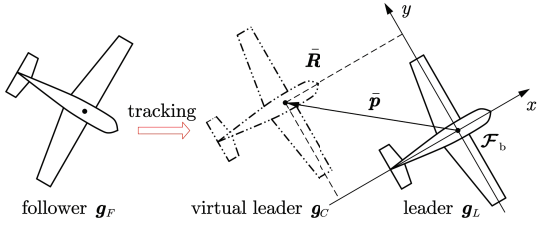}
  \caption{Leader-follower formation with the aid of a virtual leader}
  \label{fig_form_track}
\end{figure}

\begin{remark}
  As shown in Figure~\ref{fig_form_track}, the virtual leader $\bm{g}_C$ given by (\ref{eq_def_virtual_leader_confi}) is defined in the body-fixed frame of the real leader $\bm{g}_L$. Since the formation pattern $\bar{\bm{g}}$ is fixed, the virtual leader will be rigidly attached to the real leader and move along with it. In this way, $\bar{\bm{g}}$ is achieved if the follower $\bm{g}_F$ tracks the virtual leader $\bm{g}_C$.
\end{remark}

Following Lemma~\ref{lem_form_2_track}, the problem to be solved is transformed to trajectory tracking. To this end, we define the relative configuration of $\bm{g}_F$ with respect to $\bm{g}_C$, that is
\begin{equation}\label{eq_g_CF}
  \bm{g}_{CF}=\bm{g}_C^{-1}\bm{g}_F.
\end{equation}
It is obtained that the time derivative of $\bm{g}_{CF}$ is
\begin{equation}\label{eq_dot_g_CF_intial}
  \dot{\bm{g}}_{CF} = \bm{g}_{CF}\bm{\xi}^{\wedge}_{CF},
\end{equation}
where $\bm{\xi}^{\wedge}_{CF}$ is the relative velocity defined by
\begin{equation}\label{eq_xi_CF}
  \bm{\xi}^{\wedge}_{CF}=\bm{\xi}_F^{\wedge}-{\rm Ad}_{\bm{g}_{CF}^{-1}}\bm{\xi}_C^{\wedge}.
\end{equation}
\begin{lemma}\label{lem_track_2_stab}
  The follower $\bm{g}_F$ tracks the virtual leader $\bm{g}_C$ if the relative configuration $\bm{g}_{CF}$ is stabilized to the identity matrix $\bm{I}_4$.
\end{lemma}
\begin{proof}
  The proof is trivial based on (\ref{eq_gF_gC}) and (\ref{eq_g_CF}).
\end{proof}

Then, the control task is further converted to stabilization of the relative configuration $\bm{g}_{CF}$ by designing the relative velocity $\bm{\xi}_{CF}^{\wedge}$. The following lemma provides a stabilization control with logarithmic feedback.

\begin{lemma}[\cite{Bullo1995Proportional}]\label{lem_log_feedback}
  Consider the system $\dot{\bm{g}}=\bm{g}\bm{\xi}^{\wedge}$ in the Lie group ${\rm SE(3)}$ and let $k_p>0$ be a control gain. Then, the control law $\bm{\xi}^{\wedge}=-k_p\log_{\rm SE(3)}(\bm{g})$ almost globally stabilizes the state $\bm{g}$ at $\bm{I}_4$ from the initial condition $\bm{g}(0)=(\bm{R}(0),\bm{p}(0))$ satisfying ${\rm tr}(\bm{R}(0))\ne -1$, where $\log_{\rm SE(3)}$ is the logarithmic map in the Lie group $\rm SE(3)$.
\end{lemma}

Based on Lemma~\ref{lem_log_feedback}, the relative velocity $\bm{\xi}^{\wedge}_{CF}$ can be designed to be $\bm{\xi}^{\wedge}_{CF}=-k_p\log_{\rm SE(3)}(\bm{g}_{CF})$. By substituting it into (\ref{eq_xi_CF}), we can obtain the formulation of the follower's velocity, that is
\begin{equation}\label{eq_xi_F_v1}
  \bm{\xi}_F^{\wedge}=-k_p\log_{\rm SE(3)}(\bm{g}_{CF})+{\rm Ad}_{\bm{g}_{CF}^{-1}}\bm{\xi}_C^{\wedge}.
\end{equation}
Considering the fact that the virtual leader is regarded merely as a middle variable, we are supposed to cancel the terms related with the configuration $\bm{g}_C$ and velocity $\bm{\xi}_C^{\wedge}$. By substituting (\ref{eq_def_virtual_leader_confi}) and (\ref{eq_def_virtual_leader_velocity}) into (\ref{eq_xi_F_v1}), it follows that
\begin{equation*}
  \bm{\xi}_F^{\wedge} = -k_p\log_{\rm SE(3)}\left((\bm{g}_L\bar{\bm{g}})^{-1}\bm{g}_{F}\right)+{\rm Ad}_{(\bm{g}_{F}^{-1}\bm{g}_{C}\bar{\bm{g}})}{\rm Ad}_{\bar{\bm{g}}^{-1}}\bm{\xi}_L^{\wedge}.
\end{equation*}
For simplicity, we define a relative configuration
\begin{equation}\label{eq_g01_def}
  \bm{g}_{LF}=\bm{g}_L^{-1}\bm{g}_F,
\end{equation}
and then $\bm{\xi}_F^{\wedge}$ can be further expressed as
\begin{equation}\label{eq_xi_F_v2}
  \bm{\xi}_F^{\wedge}=-k_p\log_{\rm SE(3)}(\bar{\bm{g}}^{-1}\bm{g}_{LF})+{\rm Ad}_{\bm{g}_{LF}^{-1}}\bm{\xi}_L^{\wedge}.
\end{equation}

Although we have derived the control input in (\ref{eq_xi_F_v2}), it cannot directly serve as the formation controller. This is because (\ref{eq_xi_F_v2}) is only applicable to fully actuated systems, while the fixed-wing UAVs are restricted by the nonholonomic constraints. Thus, in the following, we will make the control input become nonholonomic constrained.

For simplicity, we refer to (\ref{eq_xi_F_v2}) as the standard velocity, and let a vector $\bm{\Xi}\in\mathbb{R}^6$ denote it, that is
\begin{equation}\label{eq_def_XI_upper}
  \bm{\Xi}=\begin{bmatrix}
             \bm{\Omega} & \bm{\Lambda}
           \end{bmatrix}^T\triangleq
           \left(-k_p\log_{\rm SE(3)}(\bar{\bm{g}}^{-1}\bm{g}_{LF})+{\rm Ad}_{\bm{g}_{LF}^{-1}}\bm{\xi}_L^{\wedge}\right) ^{\vee},
\end{equation}
where $\bm{\Omega}=[\Omega^x\ \ \Omega^y\ \ \Omega^z]^T\in\mathbb{R}^3$ and $\bm{\Lambda}=[\Lambda^x\ \ \Lambda^y\ \ \Lambda^z]^T\in\mathbb{R}^3$ can be regarded as the standard angular velocity and linear velocity, respectively. Due to the nonholonomic constraints, the linear velocities $v^y$ and $v^z$ of the fixed-wing UAVs are always zero. Then, the standard velocity components $\Lambda^y$ and $\Lambda^z$ derived from (\ref{eq_def_XI_upper}) cannot be provided to the fixed-wing UAVs through the input channels of $v^y$ and $v^z$. That is to say, the real linear velocity of the fixed-wing UAVs is actually
\begin{equation}\label{eq_Lambda_NH}
  \bm{\Lambda}_{NH}=\begin{bmatrix}
                 \Lambda^x & 0 & 0
               \end{bmatrix}^T.
\end{equation}
Note that the nonholonomic constraints restrict the direction of the linear velocity. This motivates us that if the direction of the real linear velocity $\bm{\Lambda}_{NH}$ is tuned to be aligned with that of the standard linear velocity $\bm{\Lambda}$, then the fixed-wing UAVs can be controlled by $\bm{\Xi}$ given in (\ref{eq_def_XI_upper}) to accomplish the formation task. Therefore, inspired by \cite{He2022Exponential}, we construct new feedback terms by employing $\Lambda^y$ and $\Lambda^z$, and let these terms be additional angular velocities, which can rotate the real linear velocity $\bm{\Lambda}_{NH}$ to the standard one $\bm{\Lambda}$ indeed. The details of the design process are given below.

We firstly construct a rotation matrix $\bm{R}_z$ around the $z$-axis of the body-fixed frame $\bm{\mathcal{F}}_{\rm b}$. Define a vector $\bm{n}=[\Lambda^x\ \ \Lambda^y\ \ 0]^T$, and it can be verified that
\begin{equation*}
  \bm{n}\cdot\bm{e}_3=0,
\end{equation*}
where the symbol ``$\cdot$" represents the dot-product and $\bm{e}_3=[0\ \ 0\ \ 1]^T$. Then, we define another vector $\bm{n}^{\bot}$ by the cross-product of $\bm{e}_3$ and $\bm{n}$, that is
\begin{equation*}
  \bm{n}^{\bot}=\bm{e}_3\times\bm{n}=\begin{bmatrix}
                                       -\Lambda^y & \Lambda^x & 0
                                     \end{bmatrix}^T.
\end{equation*}
Thus, $\{\bm{n},\bm{n}^{\bot},\bm{e}_3\}$ constitute a set of orthogonal vectors in $\mathbb{R}^3$. Based on such a group of vectors, we can construct the following orthogonal matrix
\begin{equation}\label{eq_def_Rz}
  \bm{R}_z=\begin{bmatrix}
             \frac{\bm{n}}{\|\bm{n}\|} & \frac{\bm{n}^{\bot}}{\|\bm{n}^{\bot}\|} & \bm{e}_3
           \end{bmatrix},
\end{equation}
which satisfies $\bm{R}_z^T\bm{R}_z=\bm{I}_3$ and ${\rm det}\bm{R}_z=1$. Therefore, $\bm{R}_z\in{\rm SO(3)}$ is a rotation matrix around the $z$-axis.

Next, similarly, we can construct a rotation matrix $\bm{R}_y$ around the $y$-axis. Define a vector $\bm{m}=[\Lambda^x\ \ 0\ \ \Lambda^z]^T$, which is orthogonal to the unit vector $\bm{e}_2=[0\ \ 1\ \ 0]^T$. Then, $\bm{m}^{\bot}$ can be given by
\begin{equation*}
  \bm{m}^{\bot}=\bm{e}_2\times\bm{m}=\begin{bmatrix}
                                       \Lambda^z & 0 & -\Lambda^x
                                     \end{bmatrix}^T.
\end{equation*}
Thus, the rotation matrix around the $y$-axis can be constructed by
\begin{equation}\label{eq_def_Ry}
  \bm{R}_y=\begin{bmatrix}
             \frac{\bm{m}}{\|\bm{m}\|} & \bm{e}_2 & \frac{\bm{m}^{\bot}}{\|\bm{m}^{\bot}\|}
           \end{bmatrix}.
\end{equation}

In fact, $\bm{R}_z$ and $\bm{R}_y$ represent the rotation matrices from the vector $\bm{\Lambda}_{NH}$ to $\bm{\Lambda}$. Hence, based on the logarithm of $\bm{R}_z$ and $\bm{R}_y$, we design the following angular velocity
\begin{equation}
  \bm{\Omega}_{AD} = (\log_{\rm SO(3)}(\bm{R}_y))^{\vee}+(\log_{\rm SO(3)}(\bm{R}_z))^{\vee}. \label{eq_Omega_AD}
\end{equation}
With the help of the additional angular velocity $\bm{\Omega}_{AD}$, the formation controller is given below.

\begin{theorem}\label{theo_form_contr}
  Let $(\bm{g}_L,\bm{\xi}_L^{\wedge})$ denote the configuration and velocity of the leader, and let $(\bm{g}_F,\bm{\xi}_F^{\wedge})$ denote those of the follower. Concerning a desired formation pattern $\bar{\bm{g}}$ given by Theorem~\ref{theo_RTI_general}, design the following controller
  \begin{subequations}\label{eq_form_contr_w_v}
  \begin{align}
    \bm{\omega}_F & = \bm{\Omega} + k_a\bm{\Omega}_{AD}, \label{eq_form_contr_w} \\
    \bm{v}_F & =\bm{\Lambda}_{NH}. \label{eq_form_contr_v}
  \end{align}
  \end{subequations}
  Then, the controller given in (\ref{eq_form_contr_w_v}) makes the follower $\bm{g}_F$ realize the formation pattern $\bar{\bm{g}}$ with respect to the leader $\bm{g}_L$, where $\bm{\Omega}$, $\bm{\Omega}_{AD}$, $\bm{\Lambda}_{NH}$ are defined in (\ref{eq_def_XI_upper}), (\ref{eq_Omega_AD}), (\ref{eq_Lambda_NH}), respectively, and $k_a$ is a positive control gain.
\end{theorem}

\begin{proof}
  According to Lemma~\ref{lem_form_2_track} and Lemma~\ref{lem_track_2_stab}, the follower $\bm{g}_F$ realizes the formation pattern $\bar{\bm{g}}$ with respect to the leader $\bm{g}_L$, if the relative configuration $\bm{g}_{CF}$ converges to $\bm{I}_4$. Thus, we have to prove that the controller (\ref{eq_form_contr_w_v}) stabilizes the kinematics (\ref{eq_dot_g_CF_intial}) at $\bm{I}_4$.

  For the sake of illustration, we express the kinematics (\ref{eq_dot_g_CF_intial}) under the exponential coordinate. Referring to \cite{Bullo1995Proportional}, the exponential coordinates of $\bm{g}_{CF}$ are defined by
  \begin{equation}\label{eq_X_CF}
    \bm{X}_{CF}=(\log_{\rm SE(3)}(\bm{g}_{CF}))^{\vee}.
  \end{equation}
  Based on Lemma~4 in \cite{Bullo1995Proportional}, the time derivative of $\bm{X}_{CF}$ is
  \begin{equation}\label{eq_dot_X_CF_1}
    \dot{\bm{X}}_{CF}=[\mathcal{B}_{X}]\bm{\xi}_{CF},
  \end{equation}
  where $[\mathcal{B}_{X}]$ is a matrix related to $\bm{X}_{CF}$. Substituting (\ref{eq_xi_CF}) into (\ref{eq_dot_X_CF_1}), we have
  \begin{equation}\label{eq_dot_X_CF_2}
    \dot{\bm{X}}_{CF}=[\mathcal{B}_{X}](\bm{\xi}_F-({\rm Ad}_{\bm{g}_{CF}^{-1}}\bm{\xi}_C^{\wedge})^{\vee}),
  \end{equation}
  which is the formulation of the kinematics (\ref{eq_dot_g_CF_intial}) under the exponential coordinates $\bm{X}_{CF}$. According to the definition of $\bm{X}_{CF}$ in (\ref{eq_X_CF}), once $\bm{g}_{CF}=\bm{I}_4$, there holds $\bm{X}_{CF}=\bm{0}$. Then, what we need to prove is that the system (\ref{eq_dot_X_CF_2}) is stabilized at $\bm{X}_{CF}=\bm{0}$ by the controller $\bm{\xi}_F$ given in (\ref{eq_form_contr_w_v}).

  Based on the components in (\ref{eq_form_contr_w})(\ref{eq_form_contr_v}), the controller $\bm{\xi}_F$ can be written in a vector form  as
  \begin{align}\label{eq_xi_F_pf_1}
    \bm{\xi}_F=
    \begin{bmatrix}
      \bm{\omega}_F \\
      \bm{v}_F
    \end{bmatrix}=
    \begin{bmatrix}
      \bm{\Omega} \\
      \bm{\Lambda}
    \end{bmatrix}+
    \begin{bmatrix}
      k_a\bm{\Omega}_{AD} \\
      \bm{\Lambda}_{AD}
    \end{bmatrix},
  \end{align}
  where $\bm{\Lambda}_{AD}$ is defined by $\bm{\Lambda}_{AD}=-[0\ \ \Lambda^y\ \ \Lambda^z]^T$. Define $\bm{\Xi}_{AD}=[k_a\bm{\Omega}_{AD}\ \ \bm{\Lambda}_{AD}]^T$, and with $\bm{\Xi}$ given in (\ref{eq_def_XI_upper}), the controller $\bm{\xi}_F$ can be further written as
  \begin{equation}\label{eq_xi_F_pf_2}
    \bm{\xi}_F=\bm{\Xi}+\bm{\Xi}_{AD}.
  \end{equation}
  Substituting (\ref{eq_xi_F_pf_2}) into (\ref{eq_dot_X_CF_2}), we obtain the following closed-loop system
  \begin{align}\label{eq_dot_X_CF_3}
    \dot{\bm{X}}_{CF} &= [\mathcal{B}_{X}](\bm{\Xi}+\bm{\Xi}_{AD}-({\rm Ad}_{\bm{g}_{CF}^{-1}}\bm{\xi}_C^{\wedge})^{\vee}).
  \end{align}
  Then, substituting (\ref{eq_def_XI_upper}) into (\ref{eq_dot_X_CF_3}), it follows that
  \begin{align}\label{eq_dot_X_CF_4}
    \dot{\bm{X}}_{CF} = &[\mathcal{B}_{X}](-k_p(\log_{\rm SE(3)}(\bar{\bm{g}}^{-1}\bm{g}_{LF}))^{\vee}+({\rm Ad}_{\bm{g}_{LF}^{-1}}\bm{\xi}_L^{\wedge})^{\vee} \nonumber \\ &+\bm{\Xi}_{AD}-({\rm Ad}_{\bm{g}_{CF}^{-1}}\bm{\xi}_C^{\wedge})^{\vee})
  \end{align}
  Utilizing the definitions in (\ref{eq_g01_def})(\ref{eq_def_virtual_leader_confi})(\ref{eq_def_virtual_leader_velocity})(\ref{eq_g_CF})(\ref{eq_X_CF}), we further formulate the closed-loop system (\ref{eq_dot_X_CF_4}) to be
  \begin{align}\label{eq_dot_X_CF_5}
    \dot{\bm{X}}_{CF} = &[\mathcal{B}_{X}](-k_p\bm{X}_{CF}+({\rm Ad}_{\bm{g}_{CF}^{-1}}\bm{\xi}_C^{\wedge})^{\vee} \nonumber \\ &+\bm{\Xi}_{AD}-({\rm Ad}_{\bm{g}_{CF}^{-1}}\bm{\xi}_C^{\wedge})^{\vee}) \nonumber \\
     = &-k_p[\mathcal{B}_{X}]\bm{X}_{CF}+[\mathcal{B}_{X}]\bm{\Xi}_{AD}.
  \end{align}
  With the definition of $[\mathcal{B}_{X}]$ in \cite{Bullo1995Proportional}, it can be derived that $[\mathcal{B}_{X}]\bm{X}_{CF}=\bm{X}_{CF}$. Then, the closed-loop system (\ref{eq_dot_X_CF_5}) can be written as
  \begin{equation}\label{eq_dot_X_CF_6}
    \dot{\bm{X}}_{CF} = -k_p\bm{X}_{CF}+[\mathcal{B}_{X}]\bm{\Xi}_{AD}.
  \end{equation}
  According to the proof of Theorem~1 in \cite{He2022Exponential}, $\bm{X}_{CF}=\bm{0}$ is the exponentially stable equilibrium of the closed-loop system (\ref{eq_dot_X_CF_6}). This indicates that the follower $\bm{g}_F$ realizes the formation pattern $\bar{\bm{g}}$ with respect to the leader $\bm{g}_L$ under the controller $\bm{\xi}_F$ in (\ref{eq_form_contr_w_v}).
\end{proof}

\begin{remark}
  The additional angular velocity $\bm{\Omega}_{AD}$ in (\ref{eq_Omega_AD}) can be further understood from two aspects. On the one hand, intuitively, $\bm{\Omega}_{AD}$ is derived from the constructed rotation matrices $\bm{R}_z$ and $\bm{R}_y$, which represent the rotation from the real linear velocity vector $\bm{\Lambda}_{NH}$ to the standard velocity vector $\bm{\Lambda}$, so that $\bm{\Omega}_{AD}$ provides an extra rotation that compensates the lack of linear velocities $v^y$ and $v^z$. On the other hand, theoretically, the additional angular velocity $\bm{\Omega}_{AD}$ (which describes the rotation around $x-,y-,z-$axis) is the feedback of the standard linear velocities $\Lambda^y$ and $\Lambda^z$ (which describes the translation along $y-,z-$axis), demonstrating the idea of feedback coupling in the control theory of underactuated systems. Note that $\Lambda^y$ and $\Lambda^z$ originally cannot be imported to the system due to the nonholonomic constraints. But with the help of $\bm{\Omega}_{AD}$, the underactuated directions becomes coupled with the actuated directions, leading to the fact that the underactuated states are able to be controlled by the feedback in the actuated directions.
\end{remark}

\section{Simulation examples}\label{sec_simulation}

This section provides two numerical simulation examples to verify the effectiveness of the proposed results.

\begin{figure}[ht]
  \centering
  \includegraphics[width=0.2\textwidth,trim=0 0 0 0,clip]{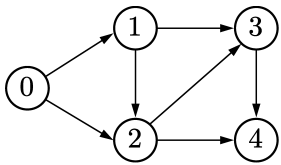}
  \caption{Communication topology of 5 fixed-wing UAVs}
  \label{fig_topol}
\end{figure}

The first example shows the formation of 5 fixed-wing UAVs interacted by a directed acyclic graph as shown in Figure~\ref{fig_topol}. The reference trajectory of the leader is chosen as two cases: a line and a helix. Note that the line reference trajectory implies the leader's angular velocity is zero, i.e., $\bm{\omega}_L=\bm{0}_{3\times 1}$. According to Corollary~\ref{theo_RTI_identity}, the desired relative rotation matrix of each UAV can be chosen as $\bar{\bm{R}}=\bm{I}_3$, indicating that all of the UAVs pointing to the same direction in the formation. Then, under the line reference trajectory, we set the desired formation pattern to be $\bar{\bm{g}}=(\bm{I}_3,\bar{\bm{p}})$, where $\bm{I}_3$ represents the rotation matrix and the position vector $\bar{\bm{p}}$ composes a wedge-shaped formation pattern. Regarding the helix reference trajectory, the position vector $\bar{\bm{p}}$ also forms a wedge shape as the above, while the rotation matrix $\bar{\bm{R}}$ is not the identity matrix but given in Theorem~\ref{theo_RTI_general}. That is to say, in this case, the orientations of the fixed-wing UAVs cannot be identical but should be decided by the nonholonomic constraints.

\begin{figure}[ht]
  \centering
  \subfigure[Line reference trajectory]{
    \includegraphics[width=0.45\textwidth,trim=5 5 5 20,clip]{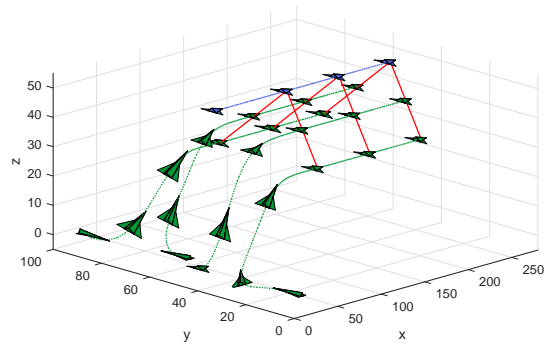}
    \label{fig_5_line_tra}}
  \subfigure[Helix reference trajectory]{
    \includegraphics[width=0.45\textwidth,trim=5 5 5 20,clip]{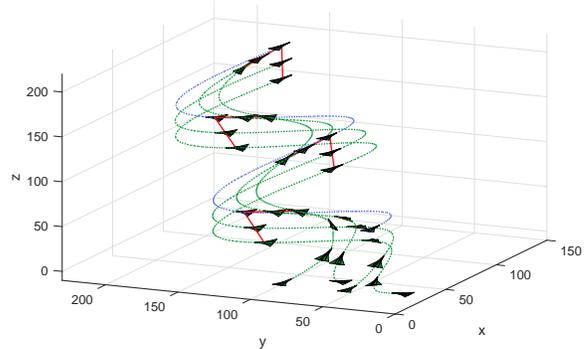}
    \label{fig_5_helix_tra}}
  \caption{Trajectories of 5 fixed-wing UAVs in RTI formation}
  \label{fig_sim_tra}
\end{figure}

Figure~\ref{fig_sim_tra} and Figure~\ref{fig_sim_confi} demonstrate the simulation results. The trajectories of the fixed-wing UAVs are shown in Figure~\ref{fig_sim_tra}, where the leader is marked in blue and the followers are in green. It can be seen from this figure, especially from Figure~\ref{fig_5_helix_tra}, that the whole formation shape can rotate and translate along with the reference trajectory, like a single rigid body, demonstrating that the formation pattern is RTI. Figure~\ref{fig_sim_confi} illustrates the relative configurations of each follower with respect to the leader, which are expressed in the leader's body-fixed frame. As we can see, the relative attitude angles $\bar{\phi},\bar{\theta},\bar{\psi}$ in these two cases are different. When the leader moves along a straight line, $\bar{\phi},\bar{\theta},\bar{\psi}$ can all be zero as shown in Figure~\ref{fig_5_line_confi}, which is guaranteed by Corollary~\ref{theo_RTI_identity}. But once the leader has nonzero angular velocity, as shown in Figure~\ref{fig_5_helix_confi}, $\bar{\theta}$ and $\bar{\psi}$ are not zero anymore, but decided by (\ref{eq_theta_feas}) and (\ref{eq_psi_feas}), respectively, which results from the nonholonomic constraints of the fixed-wing UAVs.

\begin{figure}[ht]
  \centering
  \subfigure[Line reference trajectory]{
    \includegraphics[width=0.45\textwidth,trim=5 5 5 5,clip]{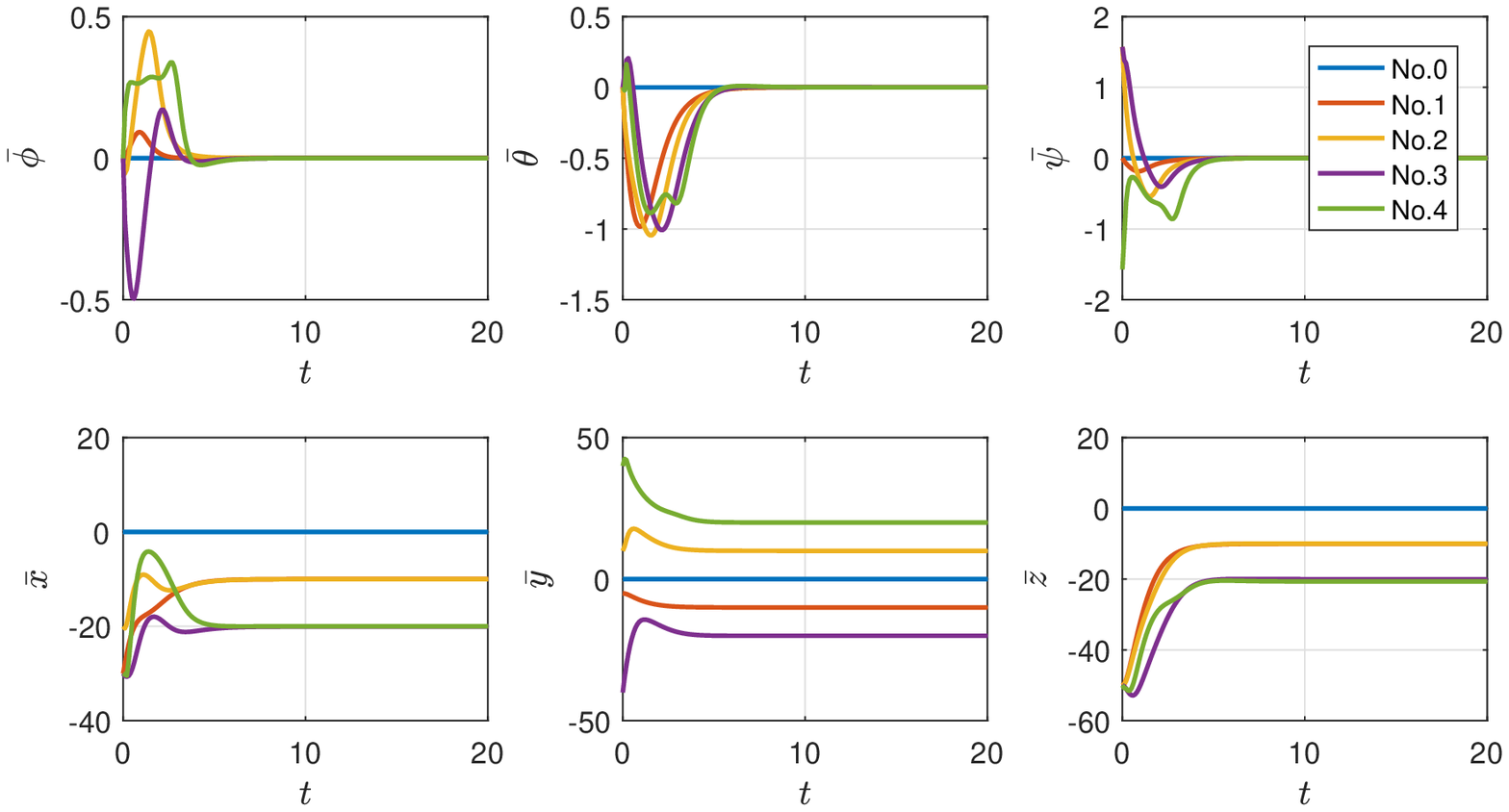}
    \label{fig_5_line_confi}}
  \subfigure[Helix reference trajectory]{
    \includegraphics[width=0.45\textwidth,trim=5 5 5 5,clip]{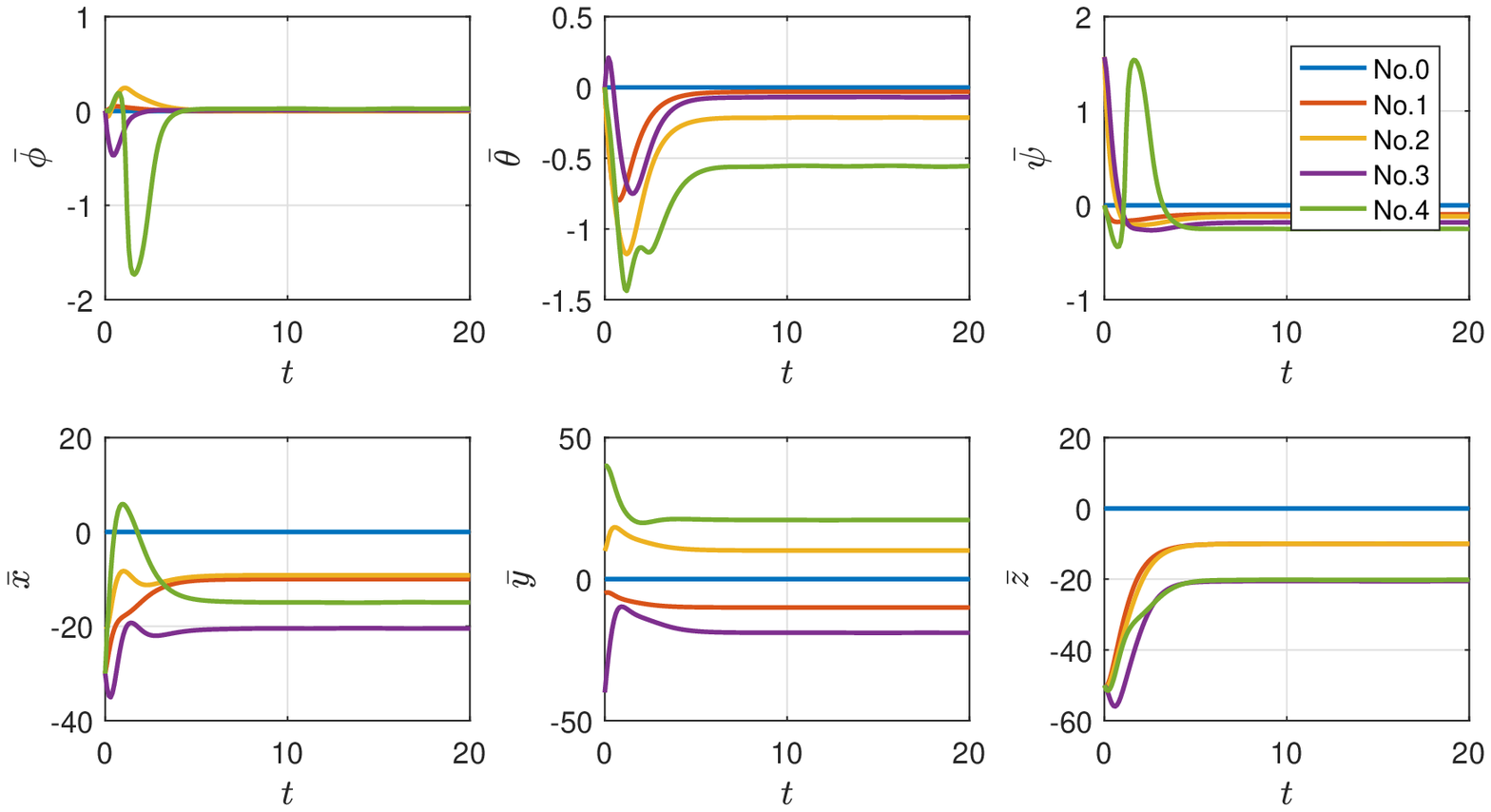}
    \label{fig_5_helix_confi}}
  \caption{Relative configurations of 5 fixed-wing UAVs in RTI formation}
  \label{fig_sim_confi}
\end{figure}

In the second example, we provide the simulation of 10 fixed-wing UAVs under a more complicated reference trajectory, where the leader's angular velocity $\bm{\omega}_L$ is a piecewise continuously differentiable function as given in Table~\ref{tab_angular}, and the leader's linear velocity $\bm{v}_L$ is set to be constant. It can be observed that such a reference trajectory contains the straight line motion, 2-D rotation and 3-D rotation. The simulation time is chosen as $T=100$s, and the results are presented in Figure~\ref{fig_sim_10}, where Figure~\ref{fig_10_tra} provides the trajectories of 10 fixed-wing UAVs in 3D, and the projections to $x-z$ plane and $y-z$ plane are given in Figure~\ref{fig_10_tra_1} and Figure~\ref{fig_10_tra_2}, respectively. It is illustrated in these figures that the formation pattern can move along the reference trajectory, exhibiting the RTI characteristics as a single rigid body.

\begin{table}[t]\small
\renewcommand{\arraystretch}{1.5}
\caption{Components of the leader's angular velocity $\bm{\omega}_L$}
\label{tab_angular}
\centering
\begin{threeparttable}
\begin{tabular}{lp{1.7cm}p{1.8cm}p{1.8cm}}
\hline
\null & $\omega_L^x$ & $\omega_L^y$ & $\omega_L^z$ \\
\hline
$0\leq t \leq 20$ & $0$                       & $-0.15\sin a(t)\tnote{*}$    & 0                            \\
$20< t \leq 30$   & $0$                       & $0$                          & $-0.25\sin b(t)\tnote{*}$    \\
$30< t \leq 50$   & $0$                       & $0$                          & $0$                           \\
$50< t \leq 55$   & $0.1\sin c(t)\tnote{*} $  & $0.15\sin c(t)\tnote{*} $    & $0.2\sin c(t)\tnote{*} $      \\
$t>55$            & $0.1$                     &  $0.15$                      &  $0.2$                          \\
\hline
\end{tabular}
\begin{tablenotes}
  \footnotesize
  \item[*] $a(t)=0.1\pi t$, $b(t)=0.1\pi (t-20)$, $c(t)=0.1\pi (t-50)$.
\end{tablenotes}
\end{threeparttable}
\end{table}

\begin{figure*}
  \centering
  \begin{minipage}[b]{0.65\textwidth}
    \subfigure[Trajectories in 3-D space]{
      \includegraphics[width=0.95\textwidth,trim=55 30 40 20,clip]{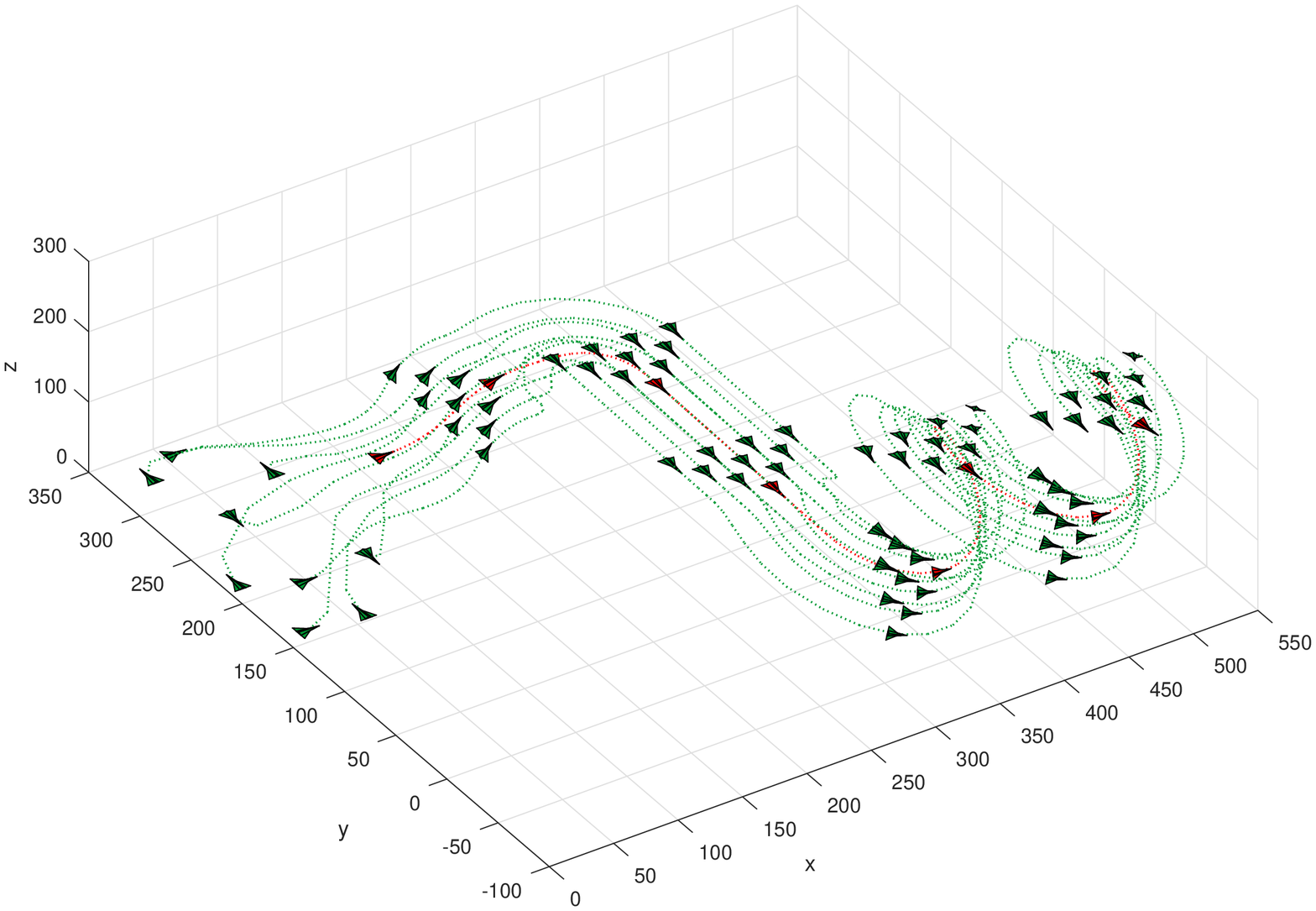}
      \label{fig_10_tra}}
  \end{minipage}
  \begin{minipage}[b]{0.3\textwidth}
    \subfigure[Trajectories in $x-z$ plane]{
      \includegraphics[width=0.95\textwidth,trim=5 5 5 20,clip]{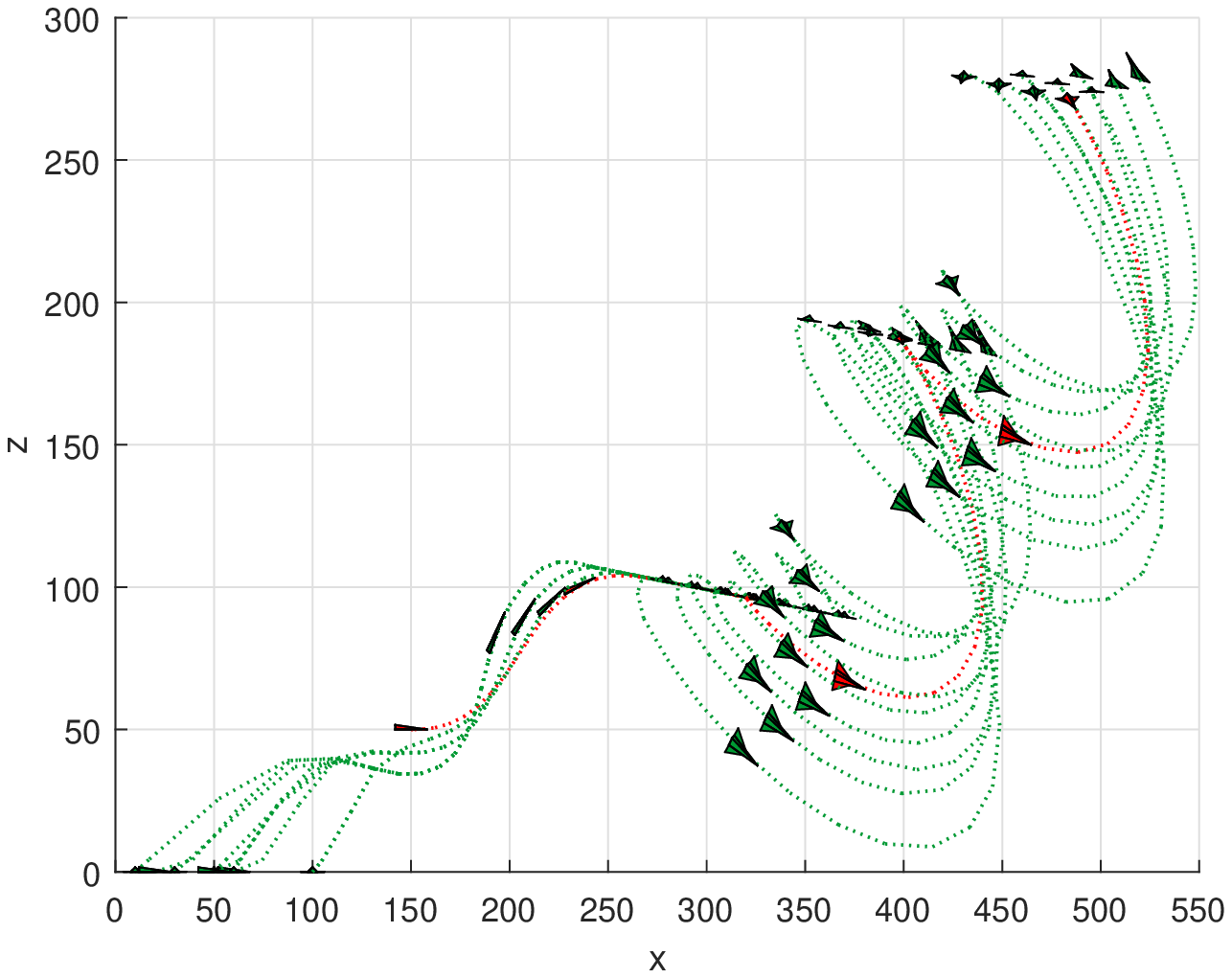}
      \label{fig_10_tra_1}} \\
    \subfigure[Trajectories in $y-z$ plane]{
      \includegraphics[width=0.95\textwidth,trim=5 3 5 20,clip]{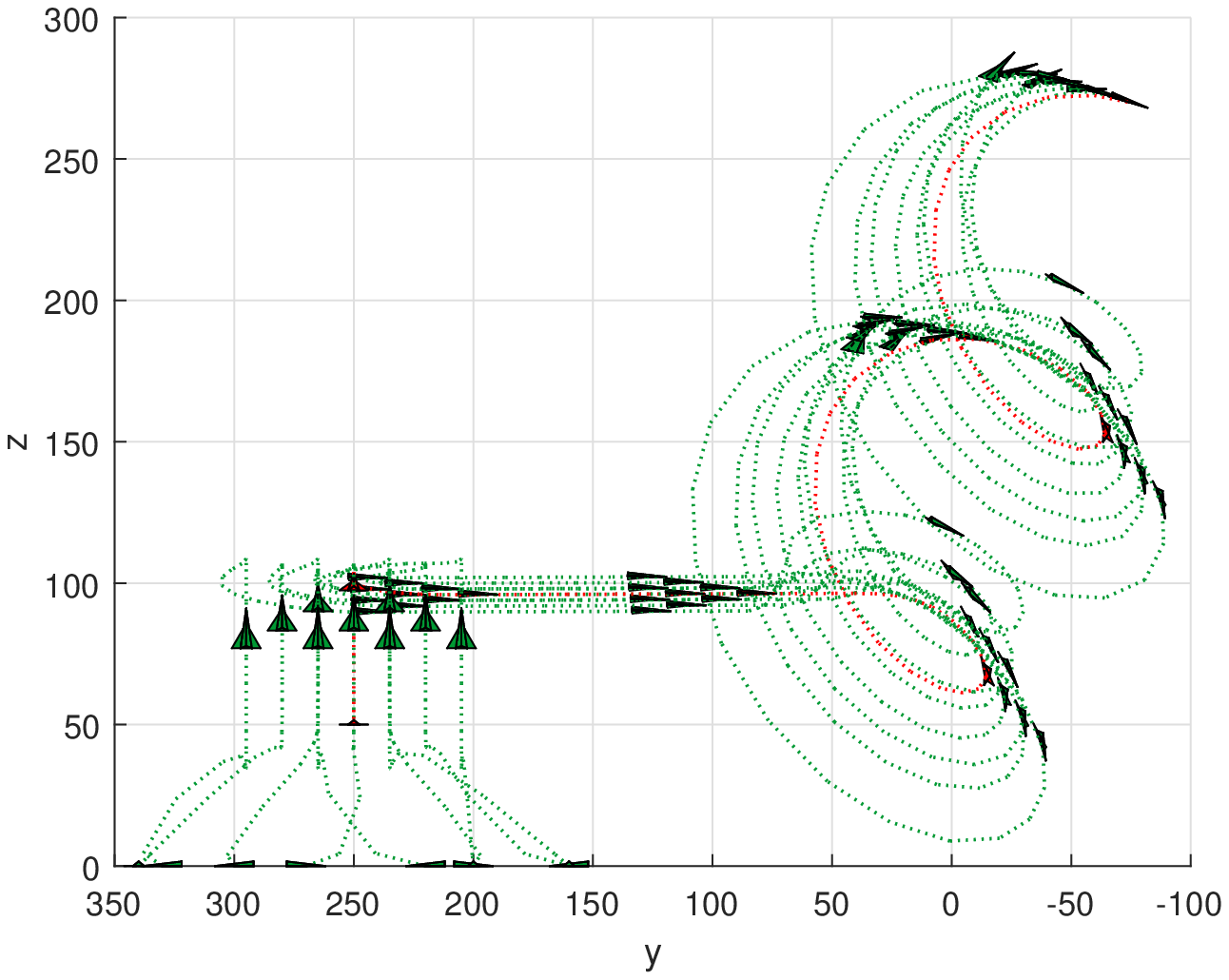}
      \label{fig_10_tra_2}}
  \end{minipage}
  \caption{Trajectories of 10 fixed-wing UAVs in RTI formation}
  \label{fig_sim_10}
\end{figure*}

\section{Conclusion} \label{sec_conclusion}

In this paper, we have investigated the RTI formation problem of the fixed-wing UAVs by proposing the formation feasibility and designing the control strategy. Particularly, the fixed-wing UAV is modelled by a rigid body in 3D, whose kinematics evolves in the Lie group ${\rm SE(3)}$. The novelty of this paper lies in the roto-translation invariance of the whole formation pattern, which characterizes an overall rigid-body motion. Furthermore, the formation feasibility has been presented under the nonholonomic and input saturation constraints, which lays the foundation to the formation control. In addition, we have employed the idea of feedback coupling in the design of control strategy so as to handle the underactuation of fixed-wing UAVs, which achieves the objective of controlling more DOFs with fewer inputs. Future works will focus on the formation problem in more practical scenarios, such as obstacle-cluttered environments, measurement noises, and so on.

\ifCLASSOPTIONcaptionsoff
  \newpage
\fi



\bibliographystyle{IEEEtran}
\bibliography{IEEEabrv,mybibfile}
\end{document}